\documentclass [11pt]{article}
\usepackage{amsmath,amsthm,amsfonts,amscd,eucal,latexsym,amssymb,mathrsfs,tikz,color,cancel}
\usetikzlibrary{decorations.markings,shapes.misc}
\usepackage{epsfig}  
\usepackage{simplewick}
\usepackage{hyperref}
\input xy
\xyoption{all}
 
\tikzset{
GFleche/.style={
 draw=black,
 postaction={decorate},
 decoration={markings,mark=at position 0.5 with {\arrow{<}}}
 },
DFleche/.style={
 draw=black,
 postaction={decorate},
 decoration={markings,mark=at position 0.5 with {\arrow{>}}}
 }}

\def\beq{\begin{eqnarray}}
\def\eeq{\end{eqnarray}}

\def\supp{\textrm{supp}}
\newcommand{\kms}[1]{\omega^{\beta}(#1)}

\def\supp{\textrm{supp}}

\newcommand{\inverse}[1]{#1^{-1}}
\newcommand{\dvol}[1]{\textrm{d}#1}

\newcommand{\translation}[1]{\alpha_{#1}}

\newcommand{\vettore}[1]{{\bf #1}}
\newcommand{\frequenza}[1]{\omega_{#1}}

\def\remark {\vsp\ifhmode{\par}\fi\noindent\noindent {\bf Remark:} 
}

\newtheorem{theorem}{Theorem}[section]

\newtheorem{proposition}[theorem]{Proposition}

\begin{document} 
%
%
 
\par 
\bigskip 
\LARGE 
\noindent 
{\bf On the stability of KMS states in perturbative algebraic quantum field theories} 
\bigskip 
\par 
\rm 
\normalsize 
 

\large
\noindent 
{\bf Nicol\`o Drago$^{1,2,a}$}, {\bf Federico Faldino$^{1,b}$}, {\bf Nicola Pinamonti$^{1,2,c}$} \\
\par
\small
\noindent$^1$ Dipartimento di Matematica, Universit\`a di Genova - Via Dodecaneso 35, I-16146 Genova, Italy. \smallskip

\noindent$^2$ Istituto Nazionale di Fisica Nucleare - Sezione di Genova, Via Dodecaneso, 33 I-16146 Genova, Italy. \smallskip
\smallskip

\noindent E-mail: 
$^a$drago@dima.unige.it, 
$^b$faldino@dima.unige.it,
$^c$pinamont@dima.unige.it\\ 

\normalsize

\par 
 
\rm\normalsize 

\rm\normalsize 
 
 
\par 
\bigskip 
 
\rm\normalsize 

\par 
\bigskip

\rm\normalsize 

\noindent
{\it  Dedicated to the memory of Rudolf Haag  
}

\bigskip

\noindent 
\small 
{\bf Abstract}.
We analyze the stability properties shown by KMS states for interacting massive scalar fields propagating over Minkowski spacetime, recently constructed 
in the framework of perturbative algebraic quantum field theories by Fredenhagen and Lindner \cite{FredenhagenLindner}.  
In particular, we prove the validity of the return to equilibrium property when the interaction Lagrangian has compact spatial support. 
Surprisingly, this does not hold anymore, if the adiabatic limit is considered, namely when the interaction Lagrangian is invariant under spatial translations.
Consequently, an equilibrium state under the adiabatic limit for a perturbative interacting theory evolved with the free dynamics does not converge anymore to the free equilibrium state. Actually, we show that its ergodic mean  converges to a non-equilibrium steady state for the free theory.
 
\normalsize

\vskip .3cm

\section{Introduction}

Recently, KMS states for interacting massive scalar quantum field theories propagating over a flat spacetime have been constructed by Fredenhagen and Lindner in \cite{FredenhagenLindner}, see also \cite{Lindner}. 
In those works, it is shown how it is possible to overcome the known infrared problems \cite{SteinmannThermal, Altherr, Landsman} when the adiabatic limit is considered. This was possible combining ideas and methods proper of statistical mechanics with new developments in the perturbative treatment of interacting field theories \cite{BF00, HW01, HW02, BFV, DF,BDF, FredenhagenRejzner, FredenhagenRejzner2}.
Actually, the methods used for $C^*-$dynamical system, see e.g. \cite{BR}, were extended to field theories treated with perturbative methods. The crucial point was the use of the time slice axiom to restrict the observables of the theory to be supported 
in a time-like compact neighborhood of a Cauchy surface in Minkowski space.  
A proof of the validity of the time-slice axiom in the context of perturbative construction of field theories has been given in \cite{CF}.
Although, in the case of field theories, it is not possible to restrict the interacting Hamiltonian to a fixed time Cauchy surface, 
it is possible to compare the free and interacting time evolution to obtain the cocycle which intertwines the two. The generator of this cocycle is a time smeared Hamiltonian.
 Having this generator at disposal the interacting KMS state can be given generalizing the construction proposed by Araki \cite{Araki} to the context of time averaged Hamiltonian. 

At first, this state is constructed for an interaction Lagrangian which has compact spatial support, this is realized multiplying the Lagrangian density by a spatial cut off function $h$. Eventually, the adiabatic limit, namely the limit where $h$ tends to $1$, is considered.  
In the case of massive theories, one obtains finite results thanks to the exponential decay of the truncated free correlation functions for large spatial directions.
If the background theory is massless and if effectively the interaction Lagrangian has a positive quadratic contribution, like the case of the thermal mass for $\lambda \phi^4$ theories, the same result can be achieved modifying the background theory by moving a quadratic contribution from the interaction Lagrangian to the free part.
 This procedure amounts to a partial resummation of the perturbative series and it can be done thanks to the  generalized principle of perturbative agreement discussed in \cite{DHP}.

In this paper, we analyze the stability of the KMS state for perturbative interacting scalar field theories constructed by \cite{FredenhagenLindner}.
In the context of $C^*-$dynamical systems, a notion of stability under perturbation was formalized by Haag, Kastler and Trych-Pohlmeyer \cite{HKT}, see also the review presented in \cite{Haag}. 
In particular, having an equilibrium state $\omega$ with respect to a free evolution described by a one parameter group of automorphisms $t\mapsto\alpha_t $ over a $C^*-$algebra $\mathscr{A}$, $\omega$ is said to be dynamically stable if it is sufficiently close to an equilibrium state $\omega^P$ of the perturbed dynamics $\alpha^P$, where $P=P^*\in \mathscr{A}$. More precisely, $\omega^{P}$ must be contained in the folium of $\omega$ and 
$\omega^{\lambda P} \to \omega$ for $\lambda \to 0$. In \cite{HKT} it is furthermore shown that if two states $\omega$ and $\omega^P$, which are invariant respectively under $\alpha$ and $\alpha^P$, are close in the previously discussed sense, the following {\bf stability condition} holds
\begin{equation}\label{eq:stability-condition-HKT}
\lim_{T\to\infty} \int_{-T}^T \dvol{t}\; \omega([P,\alpha_t(A)]) = 0
\end{equation}
for $A\in\mathscr{A}$.
The line of argument can also be reverted, actually in \cite{HKT} it is proved that if an $\alpha$-invariant state $\omega$ satisfies the condition \eqref{eq:stability-condition-HKT} for sufficiently many $A$ and $P$ in $\mathscr{A}$ and if further ergodicity assumptions hold, then $\omega$ is a KMS state at some inverse temperature $\beta$ with respect to $\alpha$. 

Robinson \cite{Robinson} proved that a state $\omega^P$ invariant under the perturbed dynamics $\alpha^P$ evolved with the free dynamics tends to an $\alpha-$invariant state $\omega$
\[
\lim_{t\to\pm \infty} \omega^{P}(\alpha_t(A)) = \omega(A)
\]
if and only if the 
\begin{equation}\label{eq:stability-bratteli}
\lim_{T\to\infty} \int_{-T}^T \dvol{t}\; \omega^P([P,\alpha_t(A)]) = 0
\end{equation}
for every $A\in\mathscr{A}$. Notice that the stability condition $\eqref{eq:stability-bratteli}$ can be seen as a first order version of 
\eqref{eq:stability-condition-HKT} for $\omega^P$ and $\alpha^P$, in the sense of perturbation theory.

Afterwards, Bratteli, Kishimoto and Robinson \cite{BrKiRo} showed that if a KMS states  $\omega$ with respect to the evolution $\alpha_t$ is {\bf strongly clustering}, i.e. 
\[
\lim_{t\to\pm \infty} \omega(A\alpha_t(B))=\omega(A) \omega(B),
\]
for $A,B\in\mathscr{A}$, the stability condition \eqref{eq:stability-condition-HKT} holds for every $P\in\mathscr{A}$. This shows the connection of the clustering condition and the stability condition. 

Here, we work with perturbative interacting scalar field theories propagating over a flat spacetime, the observables of the theory are thus known as formal power series in the coupling constants of the interaction Lagreangean, hence, we have to extend some of the conditions and methods described above.
We shall perform this analysis using methods proper of the recently developed perturbative algebraic quantum field theory (pAQFT) \cite{BDF,  FredenhagenRejzner, FredenhagenRejzner2}. In particular, we prove the validity of the clustering condition in the sense of formal power series for the perturbatively constructed 
KMS state $\omega^{\beta,V}$ if the interaction Lagrangian $V$ has compact spatial support. Furthermore, we show that, condition \eqref{eq:stability-bratteli} holds for $\omega^{\beta}$ evolved with the perturbed dynamics $\alpha^V$. Hence, we get the {\bf return to equilibrium} property for perturbative interacting field theories
\begin{equation}\label{eq:retrun-to-equilibrium}
\lim_{T\to\infty} \omega^\beta(\alpha^V_T(A)) =\omega^{\beta,V}(A)
\end{equation}
where $\omega^{\beta}$ is the extremal KMS state of the free theory, in close analogy to the results of Robinson \cite{Robinson}. Afterwards, we see that the strong clustering condition ceases to hold if $A$ has no compact spatial support. This has the effect of breaking the return to equilibrium property in the case of interaction Lagrangian which are constant in space, namely if the adiabatic limit is considered. Actually, in the analysis of $\omega^\beta\circ\alpha^V_T$ for large $T$ we encounter some divergences. 
Thus, in the last part of the work, we revert the point of view and we prove that the ergodic mean of $\omega^{\beta,V}\circ \alpha_T$ converges in the sense of formal power series to a state $\omega^+$ which can be seen as a {\bf non-equilibrium steady state} (NESS) \cite{Ru00} for the free theory.
However, a direct evaluation of {\bf entropy production}, originally given  in the case of $C^*-$dynamical systems by  Ojima, Hasegawa and  Ichiyanagi  \cite{Oj0, Oj1, Oj2} and by Jak{\v s}i\'c and Pillet \cite{JP01, JP02}, seems to give divergent results. We argue, in the last section, that this situation should be improved when densities are considered. 

The paper is organized as follows. In the next section, in order to fix the notation, we briefly review the construction of KMS states for interacting field theories in the framework of pAQFT. In the third section we discuss the validity of various clustering conditions for KMS states under perturbations of compact spatial support. Out of these conditions it is possible to prove the return to equilibrium \eqref{eq:retrun-to-equilibrium}.
In the forth section we show that the clustering condition does not hold under the adiabatic limit. 
In the fifth section we show how the failure of the clustering condition, can be used to construct perturbatively non-equilibrium steady states for the free theory.
Finally, some conclusions and some open questions are presented in the last section and few technical results are collected in the appendix.

\bigskip

\section{KMS states in the framework of pAQFT}

\subsection{pAQFT: functional approach for interacting scalar field theories}

In the functional approach to quantum field theories \cite{BF09,BDF,FredenhagenRejzner,FredenhagenRejzner2}, observables are described by functionals over field configurations. In the case of a scalar field theory, as the one we are considering in this paper, the {\bf field configurations} are assumed to be real smooth functions over the spacetime manifold $M$ and are indicated by $\phi\in \mathscr{C}:=C^\infty(M;\mathbb{R})$.  
In the analysis presented in this paper $M$ is the flat Minkowski spacetime whose metric $\eta$ is assumed to have the signature $(-+++)$.
Despite this fact some of the general definitions we are going to discuss hold also in the case of a generic globally hyperbolic spacetime. 
Coming back to the main issue, observables of the theory are functionals which have further regularity, in particular, they are considered to be {\bf smooth}, {\bf compactly supported} and {\bf microcausal}. 
	The set of such observables is denoted by $\mathscr{F}_{\mu c}$.

We recall briefly that a functional $F$ over $\phi\in\mathscr{C}$ is smooth if the functional derivatives $\frac{\delta^n F}{\delta \phi^n}$ at any order $n$  exist in $\mathcal{D}'(M^n)$. A smooth functional is compactly supported if all its functional derivatives are distributions of compact support. 
Finally, a smooth compactly supported functional is microcausal if and only if its wave front set is such that
\[
\text{WF}\left(\frac{\delta^n F}{\delta \phi^n}\right)\cap \{(x_1, \dots, x_n;k_1, \dots, k_n) | (k_1, \dots, k_n) \in (\overline{V}^n_+ \cup \overline{V}^n_-)_{(x_1, \dots  , x_n)}\} = \emptyset,
\]
where $\overline{V}^n_\pm $ are the closed future/past light cones in the cotangent bundle $T^*M^n$ with respect of the metric $g$. 

Among the set of all possible observables, the elements which play a special role are the local ones, namely 
		\[
		\mathscr{F}_{loc} = \left\{\left.  F\in \mathscr{F}_{\mu c} \right| \supp F^{(n)} \subseteq \text{Diag}_n     \right\}
		\] 
where $\text{Diag}_n=\{(x_1,\dots, x_n)\in M^n | x_1=\dots= x_n\}$.		
		If we equip $\mathscr{F}_{\mu c}$ with the pointwise product $F\cdot G (\phi) = F(\phi)G(\phi)$ and with a $*-$operation implemented by the complex conjugation $F^*(\phi) = \overline{F(\phi)}$ we obtain a (commutative) $*-$algebra, which can be interpreted as the $*-$algebra of classical observables.

The dynamics of free field theories is described by $P\phi = 0$, where $P$ is a linear hyperbolic differential operator which is $P=-\Box + m^2$ in the massive Klein Gordon case discussed here. At algebraic level, it can be implemented quotienting out the ideal generated by linear functionals of the form 
$F(P\phi)$. However, perturbation theory does not respect this quotient, hence we shall work with the off shell algebra. The equation of motion are eventually taken into account requiring the correlation functions of the considered state to be on shell. 
	
\bigskip	
	
	The quantization of free (linear) theories is realized by formal deformation of the pointwise product into a {\bf $\star-$product}. The obtained elements are then formal power series in $\hbar$ with values in $\mathscr{F}_{\mu c}$. In particular, the product we are considering is 
	\begin{equation}\label{eq:product}
	F\star_\omega G := e^{\hbar \langle \omega,\frac{\delta^2}{\delta \phi\delta \phi'} \rangle }\left. F(\phi) G(\phi')\right|_{\phi'=\phi}
	\end{equation}
	 where $\omega$ is an {\bf Hadamard bi-distribution}, namely the two-point function of an Hadamard state, associated with the free dynamics. We recall that the two-point function is Hadamard if its wave front set (WF) is microcausal, if its antisymmetric part is $i/2$ times the causal propagator and if it solves the linear equation of motion up to smooth terms \cite{Radzikowski}, see also \cite{BFK}. In this way, we have that the canonical commutation relations are satisfied also in the off shell algebra. In particular, considering {\bf local linear fields}, which are functionals of the form $F_f(\phi):=\int \phi f d\mu_g$ where $f\in \mathcal{D}(M)$ and thus they are linear in the field configurations, it holds that 
	 \[
	 [F_f,F_h]_\star := F_f \star F_h - F_h \star F_f = i\hbar \Delta(f,h), \qquad f,h\in \mathcal{D}(M)
	 \]  
	where $\Delta$ is the canonical commutator function namely the advanced minus retarded fundamental solutions of the operator $P$. Hence, the formal power series in $\hbar$ with coefficient in $\mathscr{F}_{\mu c}$, forms the algebra of the observables of the free quantum theory\footnote{The set of formal power series with coefficients in $\mathscr{F}_{\mu c}$ are usually denoted by $\mathscr{F}_{\mu c}[[\hbar]]$, if not strictly necessary we shall denote them as $\mathscr{F}_{\mu c}$.}.

The Hadamard bidistributions are uniquely determined by the WF condition only up to smooth bidistributions. However, this ambiguity in the definition of the $\star$ product is harmless in the construction of the algebra of free quantum observables because different algebras are $*-$isomorphic. The $*-$isomorphism of $(\mathscr{F}_{\mu c},\star_\omega)$ to $(\mathscr{F}_{\mu c},\star_{\omega+w})$ is realized by 
\begin{equation}\label{eq:gamma}
\gamma_w(F) := e^{\hbar \langle w,\frac{\delta^2}{\delta \phi^2} \rangle } F.
\end{equation}
Local functionals which are not linear in the field configurations, are not invariant under the action of this isomorphisms, hence local observables acquire different meaning depending on the choice of $\omega$. This freedom is covered by usual renormalization freedom present in the definition of Wick polynomials \cite{HW01,HW02}.

\bigskip

In order to construct perturbatively the interacting algebra of fields we have to consider another product. Namely, time-ordered products among local fields. 
The time order product of local fields is a map $T: \mathscr{F}_{loc}^{\otimes n} \to \mathscr{F}_{\mu c}$ which satisfies certain axioms
\cite{BF00, HW02}. A detailed analysis can be found in \cite{HW02}, here we would like to discuss only the {\bf causal factorization} property which says that 
\[
T(A,B) = T(A)\star T(B), \qquad \text{if}\qquad A\gtrsim B, \qquad T(A,B)=T(B)\star T(A) , \qquad \text{if} \qquad B\gtrsim A
\]
where $A\gtrsim B$ holds if there exists a Cauchy surface $\Sigma$ such that $\supp(A)\subseteq J^+(\Sigma)$ and $\supp(B)\subseteq J^-(\Sigma)$.
This property fixes the time ordering of $n$ local fields everywhere in $M^n$ up to the diagonals.  
Epstein-Glaser \cite{EG} recursive  construction of the time ordered product together with Steinmann \cite{Steinmann} scaling limit techniques to extend the obtained distribution on the full diagonal at each recursive step furnishes a concrete realization of the time ordered products, see e.g. \cite{BF00, HW02}.
With the time ordered products at disposal we are able to construct the time ordered exponentials of a local field $V$ as a formal power series in the coupling constant $\lambda$
\[
S(V):= \sum_n 
\frac{1}{n!}\left(\frac{i\lambda}{\hbar}\right)^n
T\left( \underbrace{T^{-1}(V)\otimes \dots \otimes T^{-1}(V)}_{n}\right)\,,
\]
The generators of the algebra of interacting fields are represented on the free algebra of fields by means of the {\bf Bogoliubov map}, namely by taking the functional derivative of the {\bf relative $S$ matrix} $\mathscr{S}_V(F)$, i.e. 
\begin{equation}\label{eq:bogoliubov-map}
\mathscr{R}_V (F) := 
\frac{\textrm{d}}{\textrm{d}\zeta}\mathscr{S}_V(F)\bigg|_{\zeta = 0}:=
\frac{\textrm{d}}{\textrm{d}\zeta}S(V)^{-1}\star S\left(V+\frac{\hbar \zeta}{i\lambda}  F\right)\bigg|_{\zeta = 0}=
S(V)^{-1}\star
T(S(V),F)\,.
\end{equation}
This map is well defined when $F$ is a local functional or a time ordered product of local functionals and thus it is well defined on the generators of the interacting algebra. Hence, even if it cannot be inverted on $\mathscr{F}_{\mu c}$, it can be used to represent the interacting $*-$algebra in the free one.
We shall thus consider the functional of the interacting algebra as the subalgebra $\mathscr{F}_I\subset \mathscr{F}_{\mu c}$ generated by the image of the local functionals under the Bogoliubov formula.
In other words, to a generic interacting local functional $F_I$ we associate $\mathscr{R}_V (F_I)$ (or $\mathscr{S}_V(F)$) which is an element of $\mathscr{F}_{\mu c}$.
	The set $\mathscr{R}_V(\mathscr{F}_{loc})\subset{\mathscr{F}_{\mu c}}$ (or the set $\mathscr{S}_V(\mathscr{F}_{loc})$) of all possible local interacting fields generates a representation of the interacting algebra over the free one.

\bigskip	 
States for the interacting algebra are constructed prescribing the form of the correlation functions among local interacting fields. It is thus sufficient to consider
\[
\omega^I(F_1,\dots, F_n)  := \omega(\mathscr{R}_V (F_1)\star\dots \star \mathscr{R}_V (F_n)) , \qquad F_i\in \mathscr{F}_{loc}. 
\]	
We conclude this section recalling that the {\bf time evolution} of the observables of the free algebra is described by a one parameter group of $*-$automorphisms $\alpha_t$. In the case of fields propagating over a Minkowski spacetime, the action of $\alpha_t$ on an element is 
\[
\alpha_t (F)(\phi) :=  F_t(\phi):= F(\phi_t), \qquad \phi_t(\tau,\vettore{x}) := \phi(\tau+t,\vettore{x})
\]
where we used standard Minkowski coordinates, adapted to $\alpha_t$, to parametrize the point $x=(t,\vettore{x})$ of the spacetime.
The {\bf interacting time evolution} is thus obtained on $\mathscr{F}_I$ employing again the Bogoliubov map. In particular, its action on a generator $\mathscr{R}_V (F)$ of $\mathscr{F}_I$ with $F\in\mathscr{F}_{loc}$ is obtained by pullback 
\[
\alpha_t^V (\mathscr{R}_V (F)) :=  \mathscr{R}_V (\alpha_t F).
\]
In the rest of the paper, we shall set $\hbar=\lambda=1$. In few cases, namely when it is necessary to fix the order in perturbation theory, we shall keep the coupling constant $\lambda\neq 1$.

\subsubsection{Adiabatic limit and local interacting potential}\label{se:cutoffs}

On Minkowski spacetime the interacting Lagrangian density $\mathcal{L}_I$ or the interacting Hamiltonian density $\mathcal{H}_I$ of the theory we would like to construct should be invariant under space and time translations\footnote{We say that $\mathcal{L}_I$ is invariant under translations if $\mathcal{L}_I(\phi)(x+y)=\mathcal{L}_I(\phi_y)(x)$ where $\phi_y(x) = \phi(x+y)$.}. We recall that
\[
V = \int \mathcal{L}_I \dvol{\mu} = -  \int \mathcal{H}_I \dvol{\mu}.
\]
where $d\mu$ denotes the measure induced by the metric $\eta$.
Unfortunately, this is in contrast with the request for $V$ to be a local field, because of the non-compact support of such $V$s. 

This clash can be overcome by means of a version of the adiabatic limit originally introduced by Hollands and Wald in \cite{HW03}. The starting observation is in the fact that, because of the time slice axioms, which holds also for perturbative theories  \cite{CF}, the algebra of observables, and thus a state, is completely determined once the expectation values of observables supported in a neighborhood of a Cauchy surface are known. 
We can thus choose $\Sigma_\epsilon=\{(t,{\bf x})| -\epsilon <t<\epsilon\}$ an $\epsilon-$neighborhood of the $t=0$ Cauchy surface, hence, the algebra of interacting fields is generated by $\mathscr{S}_V(F)$ with $F$ supported on $\Sigma_\epsilon$.
Because of the causal factorization property, it is enough to have an interacting Lagrangian which is constant on $\Sigma_\epsilon$. 

The causal factorization property satisfied by the $S$ matrix permits to restrict the support of the interaction potentials in a slightly larger neighborhood, say $\Sigma_{2\epsilon}$. Consider $V_g = \int g \mathcal{L}_I d\mu$ with $g\in \mathcal{D}(M)$. If $F$ is supported in $\mathcal{O}$, $\mathscr{S}_{V_g}(F)$ is equal to $\mathscr{S}_{V_{g'}}(F)$ whenever $g-g'$ is supported outside $J^-(\mathcal{O})$. Furthermore, if $g-g'$ is supported outside $\mathcal{O}$ the difference of the two relative $S-$matrices is taken into account by a unitary operator, actually,  
\[
\mathscr{S}_{V_{g'}}(F) = W(g',g)\star\mathscr{S}_{V_g}(F)\star W(g',g)^{-1}, \qquad 
 W(g',g) := \mathscr{S}_{V_g}(V_g-V_{g'})^{-1}.
 \]
This observation permits to construct the algebra of interacting observables by means of the algebraic adiabatic limit of $g\to 1$, namely, considering some sequence of compactly supported cut off functions $g_i$ which are equal to one on larger and larger compact regions \cite{BF00}. 
Even if this construction is always well defined at the algebraic level, the existence of a state in the limit needs to be discussed case by case.

Following \cite{FredenhagenLindner}, in order to construct states of quantum interacting scalar field theories in the adiabatic limit, 
we thus proceed that way, first of all we consider an interacting potential which is supported on compact region introducing both a temporal cutoff $\chi$ and a spatial cutoff $h$ such that $g(t,\vettore{x})=\chi(t)h(\vettore{x})$ and eventually we remove the spatial cutoff $h$ taking the limit $h\to 1$. 
Since the algebra we are considering is supported in $\Sigma_\epsilon$ an $\epsilon-$neighborhood of the Cauchy surface $\Sigma$, the temporal cutoff $\chi$ is a time compact smooth function chosen to be equal to $1$ on $\Sigma_\epsilon$ and $0$ outside $\Sigma_{2\epsilon}$, while $h$ is space compact functions equal to $1$ on a region where the considered observables are supported.
The potential we are thus considering is of the form
\begin{equation}\label{eq:potential-cutoffs}
V_{\chi, h} = \int \chi h \mathcal{L}_I \dvol{\mu}.
\end{equation}
In order to simplify the notation, we shall often skip the subscripts $\chi$  and $h$. 

The adiabatic limit discussed so far gives a well defined state, if the correlation functions of the theory converge.  Due to the spatial decay properties of the KMS state of the free massive Klein Gordon theory this construction furnishes a well defined KMS state \cite{FredenhagenLindner}. The same procedure can be employed for $\lambda\phi^4$ theories over massless background combining the principle of perturbative agreement with the observation that an effective thermal mass is present in $V$, see also \cite{DHP}.

\subsection{KMS states and the adiabatic limit}

Following \cite{FredenhagenLindner}, when the interaction Lagrangian is of the form \eqref{eq:potential-cutoffs}, a KMS state of an interacting field theory can be constructed starting from $\omega^\beta$, the extremal KMS state of inverse temperature $\beta$ with respect to $\alpha$ of the free theory, as
\begin{equation}\label{eq:int-impl}
\omega^{\beta,V}(A) = \frac{\omega^\beta(A \star U_V(i\beta))}{\omega^\beta (U_V(i\beta))}
\end{equation}
which is well defined because
\[
t\mapsto \frac{\omega^\beta(A \star U_V(t))}{\omega^\beta (U_V(t))}
\]
is analytic in some strip so the analytic continuation is well defined.
Furthermore, $U_V$ is the intertwiner between the free and the interacting dynamics and it is constructed for every $t>0$ as
\begin{equation}\label{eq:intertwiner}
U_V(t) := 1+\sum_{n=1}^\infty i^n\int_0^t\dvol{t_1} \int_0^{t_1}\dvol{t_2}\dots \int_0^{t_{n-1}}\dvol{t_n}\; \alpha_{t_n}(K_h^\chi)\star \dots \star\alpha_{t_1}(K_h^\chi).
\end{equation}
The previous formula is obtained starting from the generator
\begin{equation}\label{eq:time-average-hamiltonian}
K_h^\chi := \mathscr{R}_V(H(h\dot\chi_-)), \qquad  H(h\dot\chi_-)   := \int h \dot\chi_- \mathcal{H}_I    \dvol{\mu} 
\end{equation}
where $\dot\chi_-$ is equal to $\dot\chi$ in the past of $\Sigma_\epsilon$ and it is set to be equal to $0$ in the future of $\Sigma_\epsilon$.
In \cite{FredenhagenLindner}, this generator was obtained comparing the free dynamics $\alpha_t$ and the interacting dynamics  $\alpha_t^V$ for small times and showing that $U_V$ constructed as in \eqref{eq:intertwiner} is the intertwiner of both dynamics at every time
\begin{equation}\label{eq:intertwine-alphas}
\alpha^V_t(A) = U_V(t)\star\alpha_t(A)\star U_V(t)^{-1}
\end{equation}
for $A\in\mathscr{F}_{\mu c}(\mathcal{O})$ with $\mathcal{O} \subset \Sigma_\epsilon$, as in section \ref{se:cutoffs}.
Furthermore, $U_V$ in \eqref{eq:intertwiner} satisfies the following cocycle condition
\begin{equation}\label{eq:cocycle}
U_V(t+s)=U_V(t) \star\alpha_t(U_V(s)).
\end{equation} 
Notice that $U_V$ is a formal power series where also $K_h^\chi$ is a formal power series. It can be seen as the time evolution operator in the interaction picture.
Furthermore, having the time averaged Hamiltonian \eqref{eq:time-average-hamiltonian} at disposal, the interacting dynamics can be explicitly written as
\begin{equation} \label{eq:pert-dyn}
\alpha_t^V(A)=
\alpha_t(A)+
\sum_{n\geq 1}i^n\int_{t S_n}\dvol{t_1}\dots\dvol{t_n}[\alpha_{t_1}(K_h^\chi),[\ldots [\translation{t_n}(K_h^\chi),\alpha_t{(A)}] \ldots]]
\end{equation}
where as usual $S_n:=\{(t_1,\dots,t_n)\in \mathbb{R}^n| 0\leq t_1\leq \dots \leq t_n\leq 1\}$ denotes the $n$-dimensional simplex.

We recall that the expectation value on $\omega$ of the $\star$ product of $n$ elements $F_i\in \mathscr{F}_{\mu c}$ can be decomposed on the {\bf connected part} $\omega^c$ of $\omega$ as follows
\begin{equation}\label{eq:connected-part}
\omega(F_1\star \dots   \star F_n) = \sum_{P\in \text{Part} \{1, \dots, n\}}  \prod_{I\in P} \omega^c\left(\bigotimes_{i\in I} F_i \right),
\end{equation}
where $\text{Part} \{1, \dots, n\}$ indicates the set of partitions of $\{1,\dots, n\}$ and the elements in the tensor product are ordered preserving the order of $\{1,\dots, n\}$. 
The connected part $\omega^c$ of $\omega$ is thus a map from the tensor algebra over $\mathscr{F}_{\mu c}$ to $\mathbb{C}$.

Employing the connected part $\omega^{\beta,c}$ of $\omega^\beta$, the interacting KMS state \eqref{eq:int-impl} can be written as a sum of integrals over $n-$dimensional simplexes $S_n$
\begin{gather}
\label{eq:interacting-KMS}
\omega^{\beta,V}(A)=
\sum_{n\geq 0}(-1)^n\int_{\beta S_n}\dvol{U}\omega^{\beta,c}\left(
A\otimes\bigotimes_{k=1}^n\alpha_{iu_k}(K)
\right)
\end{gather}
where the analytic properties of the connected $n-$point functions are used to give meaning to the analytic continuation written above. In 
\cite{FredenhagenLindner}, it is shown that the limit $h\to1$ of the previous expression can be taken in the sense of van Hove.

\section{Stability and KMS condition}

We would like to study the stability of the KMS state of the free theory under perturbations described by the potential $V_{\chi,h}$ in \eqref{eq:potential-cutoffs}. 
To this end we recall that $\omega^\beta$, the extremal and quasi-free KMS state of the free theory with respect to the time evolution $\alpha_t$ at inverse temperature $\beta$, is a quasi-free state completely determined by the two-point function
\begin{equation}\label{eq-KMS-free}
\omega_2^\beta(x,y)  = \frac{1}{(2\pi)^3}\int\dvol{^3\vettore{k}}\left(b_+(k)\frac{e^{i\frequenza{\vettore{k}}(t_x-t_y) }}{2\frequenza{\vettore{k}}}+b_-(k)\frac{e^{-i\frequenza{\vettore{k}}(t_x-t_y) }}{2\frequenza{\vettore{k}}}\right)e^{-i\vettore{k}\cdot(\vettore{x}-\vettore{y})}
\end{equation}
where, $\frequenza{\vettore{k}}=\sqrt{\vettore{k}^2 + m^2}$ and $b_+(k)=\inverse{(1-e^{-\frequenza{\vettore{k}}\beta})}$, $b_-(k)=e^{-\beta \frequenza{\vettore{k}}}b_+(k)$.

Notice that the free KMS state $\omega^\beta$ induces a state on the interacting algebra  $\mathscr{F}_I$ which we recall is generated by the elements of $\mathscr{F}_{\mu c}$ obtained applying the Bogoliubov map \eqref{eq:bogoliubov-map} to local fields. Notice that, in order to simplify the computation of expectation values, we shall use the representation of the free algebra where the $\star_{\omega^\beta}$ product is employed. 
We shall now check if the free KMS state $\omega^\beta$ is stable under perturbation of the dynamics at least in the asymptotic regime. This analysis is performed extending the results of Bratteli, Kishimoto and Robinson \cite{BrKiRo} to the case of perturbative quantum field theories.
Namely, we would like to check if the return to equilibrium \eqref{eq:retrun-to-equilibrium} ($\lim_{T\to\infty} \omega^\beta({\translation{T}^V(A)}) = \omega^{\beta,V}(A)$) holds.
Similar results in the context of $C^*$-algebras \cite{BrKiRo} suggest that the limit should be $\omega_\beta^V(A)$, i.e. the interacting KMS state with respect to the translation $\alpha_t^V$. The generalization of these fact needs to be carefully addressed due to the weaker convergence conditions present in our case, actually, the elements of the interacting algebras are known only as formal power series.
As a preliminary step we would like to show that $\omega^\beta$ satisfies a clustering condition. We have actually the following proposition
\begin{proposition}\label{pr:clustering} (Clustering condition for $\alpha_t$) 
Consider $A$ and $B$ two elements of $\mathscr{F}_I(\mathcal{O})$ where the interacting potential is $V_{\chi,h}$ ($h$ is of compact spatial support, $\chi=1$ on $\mathcal{O}$). 
It holds that 
\[
\lim_{t\to\infty} \omega^\beta(A\star \alpha_{t}(B)) = \omega^\beta(A)\omega^\beta(B) 
\]
in the sense of formal power series in the coupling constant.
\end{proposition}
\begin{proof}
$A$ and $B$ can be constructed as sums of star products of the form $\mathscr{R}_V(F_1)\star \dots \star \mathscr{R}_V(F_n)$ where $F_n$ are local fields. 
Since $V_{\chi,h}$ is of compact support, the product $\mathscr{R}_V(F_1)\star \dots \star \mathscr{R}_V(F_n)$ is of compact support too. 
Actually,  $\supp{(\mathscr{R}_{V_{\chi,h}}(F_i))} \subset J^+(\supp V_{\chi,h}) \cap J^-(\supp F_i)$ which is a compact set if  $V_{\chi,h}$ is of spatial compact support  and if $F_i\in \mathscr{F}(\mathcal{O})$. 
Hence, the supports of both $A$ and $B$ are compact. Suppose for simplicity that both of them are contained in a compact set $\mathcal{C}$.
From \eqref{eq:product} and thanks to the invariance of $\omega^\beta$ under $\alpha$, we have
\begin{equation}\label{eq:int}
\omega^\beta(A\star \alpha_{t}(B)) -  \omega^\beta(A)\omega^\beta(B) = \sum_{n\geq 1}\frac{1}{n!} \langle A^{(n)}   ,{\omega_2^\beta}^n {(\alpha_t(B))}^{(n)} \rangle_n\;.
\end{equation}
Notice that the distributional support of $A^{(n)}$ is contained in $\mathcal{C}^n$ while that of ${(\alpha_t(B))}^{(n)}$  is contained in 
$\mathcal{C}_t^n$ where $\mathcal{C}_t$ is $\mathcal{C}$ translated in time $t$, namely
\[
\mathcal{C}_t:=\{(\tau,\vettore{x})\in M |   (\tau-t,\vettore{x})\in \mathcal{C}\}.
\]
Whenever $t$ is sufficiently large  there are no null geodesics which intersect at the same time $\mathcal{C}$ and $\mathcal{C}_t$. For this reason, and because  $(x,y)$ are contained in the singular support of $\omega_2^\beta$ only if $x,y$ are joined by a null geodesic, the integral kernel of ${\omega_2^\beta}^n$ restricted to $(\mathcal{C}\times \mathcal{C}_t)^n$ is smooth if $t$ is sufficiently large. 
Hence, the limit of \eqref{eq:int} for large times is governed by the decaying properties of the two-point function $\omega_2^\beta(x,y)$ when $x-y$ is a large timelike vector. 
More precisely, we know from Proposition \ref{pr-decay-omega-beta} stated in the appendix that 
\begin{equation}\label{eq:decay2pt}
\left|D^{(\alpha)}\omega_2^\beta(x; t_y+t, \vettore{y})\right|\leq \frac{C_\alpha}{t^{3/2}} 
\end{equation}
where $\alpha\in\mathbb{N}^8$ is a multi index and $D^{(\alpha)}$ indicates partial derivatives of order  $|\alpha|=\sum_{i=1}^n\alpha_i$ along the directions determined by $\alpha$ and $C_\alpha$ are some positive constants. Furthermore, 
\[
\left\langle A^{(n)}   ,{\omega_2^\beta}^n {(\alpha_t(B))}^{(n)} \right\rangle_n= A^{(n)}\otimes B^{(n)}(\Lambda_{\mathcal{C}^{2n}} {\omega^\beta_t}^n)
\]
where, $\omega^\beta_t$ equals $\omega_2^\beta$ with the second entry translated by $-t$ along the selected Minkowski time.
$\Lambda_{\mathcal{C}^{2n}}$ is a compactly supported function in $M^{2n}$ which is equal to $1$ on $\mathcal{C}^{2n}$, notice that the precise form of the function does not enter in the final result because of the support of $A^{(n)}\otimes B^{(n)}$. Furthermore, as discussed above, for large values of $t$, 
$\Lambda_{\mathcal{C}^{2n}} {\omega^\beta_t}^n$ is smooth. 
The distributions $A^{(n)} \otimes B^{(n)} $ are of compact support, hence, by continuity we have that
\[
\left| \left\langle A^{(n)}   ,{\omega_2^\beta}^n {(\alpha_t(B))}^{(n)} \right\rangle_n\right| \leq 
C_n \sum_{|\alpha|<K_n}  \left\| D^{(\alpha)}\Lambda_{\mathcal{C}^{2n}} {\omega^\beta_t}^n \right\|_\infty,
\]
where $C_n$ and $K_n$ are fixed constants.
Thanks to \eqref{eq:decay2pt}, the right hand side of the previous inequality vanishes in the limit $t\to\infty$.
Hence, we have the result because, the sum over $n$ in \eqref{eq:int} converges to $0$ in the sense of formal power series. 

\end{proof}

The clustering condition established in Proposition \ref{pr:clustering} permits to show that the interacting KMS $\omega^{\beta,V}$ evolved with the free evolution converges pointwise in $\mathscr{F}_I$ to the free KMS state. Actually,
\begin{equation}\label{eq:reverse-stability}
\lim_{T\to\infty} \omega^{\beta,V}(\alpha_T(A))  = \lim_{T\to\infty} \frac{\omega^{\beta}(\alpha_T(A)\star U_V(i\beta))}{\omega^\beta(U_V(i\beta))}  = \omega^{\beta}(A)
\end{equation}
where the limit is taken in the sense of perturbation theories, namely, after expanding both sides of the equality in the perturbation parameter $\lambda$ we have convergence order by order. Furthermore, in the first equality we used the definition \eqref{eq:int-impl} and in the second one the result of  Proposition \ref{pr:clustering} extended to $U_V$.
As discussed in the introduction, the previous condition \eqref{eq:reverse-stability}  is very close to one of the stability conditions analyzed in \cite{HHW, Robinson, BrKiRo}, see \eqref{eq:stability-bratteli}. However, for our purposes, we would like to prove that the free KMS state evolved with the interacting dynamics tends to the interacting KMS state constructed in \cite{Lindner, FredenhagenLindner}, i.e. the limit stated in \eqref{eq:retrun-to-equilibrium}. 
In any case, the clustering condition obtained in Proposition \ref{pr:clustering} permits to have stability up to first order in $V$. 
Actually we have the following Theorem, whose proof can be given in close analogy to the proof of Theorem 2 in \cite{BrKiRo}. 

\begin{theorem}\label{th:first-order-stability-compact}{(First order stability)}
Let $\omega^\beta$ be the extremal KMS state with respect to the evolution $\alpha_t$ at inverse temperature $\beta$ of the free theory. Then, 
under perturbations described by $V_{\chi,h}$ in \eqref{eq:potential-cutoffs},
return to equilibrium \eqref{eq:retrun-to-equilibrium} holds at first order, i.e. 
\begin{equation}\label{n=1terms-a}
\lim_{T\to\infty} i\int_0^T\dvol{t}\;\kms{[\alpha_{t}(K),\alpha_T(A)]} =
-\int_0^\beta\dvol{u}\;\omega^{\beta,c}\left(
A\otimes\alpha_{iu}(K)
\right)
\end{equation}
where $A$ is an element of $\mathscr{F}_I(\mathcal{O})$ and $K$ is as in \eqref{eq:time-average-hamiltonian} with $\chi=1$ on $J^+(\mathcal{O})$.
\end{theorem}
\noindent
Before discussing the proof, we notice that, the right hand side of \eqref{n=1terms-a} is the first order contributions in $K$ of $\eqref{eq:interacting-KMS}$ while the left hand side is that of $\omega^\beta$ composed with $\alpha^V$ in \eqref{eq:pert-dyn}. Since $K$ is itself a formal power series in the coupling constant $\lambda$ where the order zero vanishes, the proposition implies the stability up to first order in the sense of perturbation theory.
\begin{proof} {\it of Theorem \ref{th:first-order-stability-compact}}
The proof can be performed in close analogy to the proof of Theorem 2 in \cite{BrKiRo} once the clustering condition stated in Proposition \ref{pr:clustering} is established. 
Let us start noticing that 
\begin{eqnarray}
i\int_0^T\dvol{t}\;\kms{[\alpha_{t}(K),\alpha_T(A)]} & =& i\int_0^T\dvol{t}\;\kms{[\alpha_{-t}(K),A]}
\notag
\\
&=& i\int_0^T\dvol{t}\left(\kms{A\star\translation{-t+i\beta}(K)}-\kms{A\star\translation{-t}(K)}\right)
\notag
\\ 
\label{eq:change-var}
&=& \int_0^\beta\dvol{u}\left(\kms{A\star\translation{-T+iu}(K)}-\kms{A\star\translation{iu}(K)}\right),
\end{eqnarray}
where the last equality holds because of the divergence theorem. Actually, thanks to the KMS property satisfied by $\omega^\beta$, the function 
$F(z):=  \kms{A\star\translation{z}(K)}$
is analytic in the strip $\Im(z)\in[0, \beta]$. Hence $\partial_{\overline{z}}F(z)=0$, thus the integral of $F(z)$ over a closed oriented curve in the domain of analyticity vanishes. 
From the clustering condition stated in Propostion \ref{pr:clustering} we have that 
\begin{gather}\label{ergodic property}
\lim_{T\to\infty}\int_0^\beta\dvol{u}\; \kms{A\star\translation{-T+iu}(K)}= \int_0^\beta \dvol{u}\;
\kms{A}\kms{K}
\end{gather}
hence, using it in \eqref{eq:change-var} and recalling the definition \eqref{eq:connected-part}, we conclude that the limit \eqref{n=1terms-a} holds.
\end{proof}
The clustering condition established in Proposition \ref{pr:clustering} does not suffice to obtain the sought return to equilibrium to all orders in $K$.
Actually, at higher orders, due to the presence of the $\star-$product of various time translated generators, the clustering condition cannot be used to factorize their expectation values.
 We have thus to introduce the following proposition.
\begin{proposition}\label{pr:clustering2} (Clustering condition for $\alpha_t^V$)
The following clustering condition, 
\[
\lim_{t\to\pm\infty} \left[ \omega^{\beta}(A\star \alpha_{t}^V(B)) - \omega^\beta(A) \omega^\beta(\alpha_{t}^V(B))  \right] =0 ,
\]
for $A$ and $B$ in $\mathscr{F}_I(\mathcal{O})$, holds in the sense of formal power series in the coupling constant whenever the perturbation Lagrangian $V_{\chi,h}$ has compact spatial support. The same result holds also when $A=U(i\beta)\star A'$ if $A'\in \mathscr{F}_I(\mathcal{O})$.
\end{proposition}
\begin{proof}
Definition \eqref{eq:connected-part} implies that the statement of the proposition holds if, 
at every order of perturbation, the connected function
\[
\omega^{\beta,c}(A \otimes \alpha_{t}^V(B)) =  \omega^{\beta}(A\star \alpha_{t}^V(B)) - \omega^\beta(A) \omega^\beta(\alpha_{t}^V(B))  
\]
vanishes for large or negative values of $t$. Let us thus expand $\alpha_{t}^V$ as in \eqref{eq:pert-dyn}. Hence,
\[
\omega^{\beta,c}(A \otimes \alpha_{t}^V(B)) =  
\sum_{n\geq 0}i^n\int_{tS_n}\dvol{t_1}\dots\dvol{t_n}\omega^{\beta,c} (A\otimes 
[\alpha_{t_1}(K),[\ldots [\translation{t_n}(K),\alpha_t(B)] \ldots]]
).
\]
The element $n=0$ in the sum vanishes in the limit of large $t$ thanks to the clustering condition given in Proposition \ref{pr:clustering}. 
We show now that the $n-$th element of the previous sum tends to zero for $t\to\infty$ in the sense of formal power series in the coupling constant.

To this end, we notice that 
$[\alpha_{t_1}(K),[\ldots ,[\translation{t_n}(K),B]\ldots ]]$
is a sum of connected components.  In order to prove it notice that
\[
[A,B]=A\star B-B\star A = m \left(e^{D_{12}} -  e^{D_{21}} \right) A\otimes B.
\]
In the previous formula $D_{ij}$ is the functional differential operator
\begin{equation}\label{eq:differential}
D_{ij}=\left\langle \omega^\beta_2,\frac{\delta^2}{\delta \varphi_i\delta \varphi_j} \right\rangle
\end{equation}
where $\frac{\delta}{\delta \varphi_i}$ is the functional derivative with respect to the $i-$th  element in the tensor product of $n$ elements and $\omega^\beta_2$ is the thermal two-point function \eqref{eq-KMS-free}.
Furthermore, $m(A_1\otimes \dots \otimes A_n) :=A_1\dots A_n$ maps the tensor product into the pointwise product.
Since 
\[
\omega^{\beta,c} (A\otimes B) = \left. m \left(e^{D_{12}} -  1  \right)A\otimes B \right|_{(\phi_1,\phi_2)=0},
\]
we conclude that $\omega^{\beta,c} (A\otimes 
[\alpha_{t_1}(K),[\ldots [\translation{t_n}(K),B] \ldots]]
)$ is a weighted sum over the set $\mathcal{G}_{n+2}^{o,c}$ of connected oriented graphs with $n+2$ vertices. Every oriented line joining two vertices indicates the presence of a two-point function. Furthermore, 
 a single graph $G\in \mathcal{G}_{n+2}^{o,c}$ cannot contain lines with opposite orientations joining the same two vertices.
 The orientation is necessary because the two-point function is not symmetric and because of the presence of subsequent commutators. 
Indicating by $c_{G}$ the weight of the single graph $G$ it holds that 
\[
\omega^{\beta,c} (A\otimes 
[\alpha_{t_1}(K),[\ldots [\translation{t_n}(K),\alpha_t(B)] \ldots]] = \sum_{G\in \mathcal{G}_{n+2}^{o,c}} c_G \int_{tS_n}\dvol{t_1}\dots \dvol{t_n}\; F_G(t_1,\dots , t_n), 
\]
where
\[
\qquad F_G(t_1,\dots , t_n):=\left. m \prod_{l\in E(G)} D_{s(l)r(l)} \left(A\otimes \alpha_{t_1}(K)\otimes \dots \otimes \alpha_{t_n}(K) \otimes \alpha_t(B)\right) \right|_{(\phi_1,\dots \phi_{n+2})=0}
\]
where $E(G)$ is the set of the lines of $G$, $s(l)$ and $r(l)$ in $\{0,\dots, n+1\}$ indicates the source and the range of the line $l$.

Since $V=V_{\chi,h}$ is of compact spatial support, $K$, $A$ and $B$ are in $\mathcal{F}_{\mu c}$. We apply Proposition \ref{pr-decay-functional-derivatives} some times to get
\begin{gather*}
|F_G(t_1,\dots , t_n)|\leq C_1\prod_{l\in E(G)}   \frac{1}{(|t_{s(l)}-t_{r(l)}|+d)^{3/2}} \\
\leq C_2{(|t_1+d||t_2-t_1+d||t_3-t_2+d|\dots |t_n-t_{n-1}+d||t-t_n+d|)}^{-3/2}
\end{gather*}
for some $d$ and some $C_1$ and $C_2$ where in the last inequality we used the fact that the graph $G$ is connected and that the times $(t_1,\dots, t_n)\in tS_n$. 
The integral over the simplex of the right hand side of the previous chain of inequalities can now be performed. Furthermore, 
it vanishes in the limit $t\to\infty$. Thus concluding the proof.
\end{proof}

Actually, the previous proposition can be used to show the validity of a (sort of) Gell-Mann Low factorization formula for $\omega^{\beta, V}$ with respect to $\omega^{\beta}$. We stress that the failure of the return to equilibrium property in the adiabatic limit $h\to 1$ -- cf. Proposition \ref{pr:infrared-divergences} in the next section -- implies that such a formula is not preserved in the non-spacelike compact case.
Owing to the clustering condition established in Proposition \ref{pr:clustering2} we can show that the free KMS state is stable. Actually we have the following theorem. 

\begin{theorem}\label{th:stability-compact}{(Stability)}
Let $\omega^\beta$ be the extremal KMS state  with respect to the evolution $\alpha_t$ at inverse temperature $\beta$ of the free theory. Then the state is stable under perturbation described by a
$V_{\chi,h}$ which is a spatially compact interacting Lagrangian. Namely
\begin{equation}\label{eq:limit-state}
\lim_{T\to\infty}\omega^\beta({\translation{T}^V(A)}) = \omega^{\beta,V}(A)
\end{equation}
where $A$ is an element of $\mathscr{F}_I(\mathcal{O})$ where $\chi=1$ on $J^+(\mathcal{O})$.
\end{theorem}
\begin{proof}
For simplicity, in this proof, we shall write $U(t):= U_V(t)$ given in \eqref{eq:intertwiner} and we shall not write explicitly the $\star$-product. 
Let us start observing that, by time translation invariance of $\omega^\beta$ and by \eqref{eq:intertwine-alphas}, we have
\begin{gather*}
\omega^{\beta}(\alpha_T^V(A))=
\omega^{\beta}(\alpha_{-T}\alpha_T^V(A))=
\omega^{\beta}(U(-T)^{-1} A U(-T))=
\omega^{\beta}( U(-T)\alpha_{i\beta}U(-T)^{-1} \alpha_{i\beta}A)
\end{gather*}
where in the last equality we have used the KMS condition. We have furthermore used the cocycle condition \eqref{eq:cocycle} to obtain $\alpha_{-T}(U(T))^{-1}=U(-T)$.
The cocycle condition \eqref{eq:cocycle} implies also that
\begin{gather*}
\omega^{\beta}(\alpha_T^V(A))=
\omega^{\beta}(U(-T)\alpha_{-T}(U(i\beta)^{-1})U(-T)^{-1} U(i\beta) \alpha_{i\beta}A) =
\omega^{\beta}(\alpha^V_{-T}(U(i\beta)^{-1}) U(i\beta) \alpha_{i\beta}A).
\end{gather*}
Proposition \ref{pr:clustering2} gives that the limit $T\to\infty$ of $\omega^{\beta}(\alpha_T^V(A))$ is such that
\[
\lim_{T\to\infty}\omega^{\beta}(\alpha_T^V(A)) =    \omega^\beta( U(i\beta) \alpha_{i\beta}A) 
\lim_{T\to\infty}( \omega^{\beta}(\alpha^V_{-T}(U(i\beta)^{-1})) ).
\]
Notice that $\alpha_T^V(1)=1$, the state $\omega^{\beta}$ is normalized, and $\omega^\beta(U(i\beta))$ is finite to all orders in perturbation theory, hence the limit $T\to\infty$ of $\omega^{\beta}(\alpha^V_{-T}(U(i\beta)^{-1}))$ converges in the sense of perturbation theory to $\omega^{\beta}(U(i\beta))^{-1}$, which is the normalization factor of $\omega^\beta( U(i\beta) \alpha_{i\beta}A)$. Hence, we have the result.
\end{proof}

\section{Instabilities in the adiabatic limit}\label{sec:inst}

In the previous section, we have established that KMS states for field theories are stable under spatially compact local perturbations.
We shall now discuss the case where the perturbation is constant in space, namely when the adiabatic limit is considered.
We shall see that the arguments used to prove Theorem \ref{th:stability-compact} or Theorem 
\ref{th:first-order-stability-compact} do not hold after having taken the adiabatic limit, hence return to equilibrium \eqref{eq:retrun-to-equilibrium} does not hold in this case.
In order to enhance the chances of having at least a convergence limit we shall here consider an ergodic (temporal) mean of the free KMS state perturbed by $V$.  
The ergodic mean is usually introduced to tame the possible oscillation for large times.
Actually, we study the {\bf ergodic mean} of $\omega^\beta \circ \translation{\tau}^V$
\begin{gather}\label{eq:average-limit-1}
\omega^{V,+}_T(A):=\lim_{h\to1}\frac{1}{T}\int_{0}^T\omega^\beta({\translation{\tau}^V(A)})\dvol{\tau}
\end{gather}
and eventually we analyze the limit $T\to \infty$.
As in the previous section, we consider a massive theory on a Minkowski spacetime perturbed with an interaction Lagrangian $V_{\chi,h}$.
We shall see that the clustering condition fails when the adiabatic limit is considered and this failure cannot be repaired by the ergodic mean.

\begin{proposition}\label{pr:failure-clustering}
Suppose that $\left.\frac{\delta^2 V_{\chi,h}}{\delta \phi \delta \phi}\right|_{\phi=0}\neq 0$.  If the adiabatic limit ($h\to 1$) is considered, 
the clustering condition fails at first order in perturbation theory also when the ergodic mean is considered, i.e.
\[
\lim_{T\to\infty}  \lim_{h\to 1}\left(  \frac{1}{T}\int_0^{T} \dvol{t}\; \omega^{\beta}(A\star\alpha_t(K))-  \omega^{\beta}(A) \omega^\beta(K) \right)\neq 0
\]
for $A=\mathscr{R}_V(F_f)\star\mathscr{R}_V(F_g)$ where $F_f$ and $F_g$ are local linear fields smeared by $f,g\in \mathcal{D}(M)$ and $K$ is as in \eqref{eq:time-average-hamiltonian}.
\end{proposition}
\begin{proof}

Let us consider the case where $V_{h,\chi}= \frac{1}{2}\int h\chi \phi^2 \dvol{\mu}$, more complicated potentials can be treated analogously. By direct computation, we get
\begin{gather*}
\omega^{\beta}(\mathscr{R}_V(F_f)\star\mathscr{R}_V(F_g)\star\alpha_t(K)) -  \omega^{\beta}(\mathscr{R}_V(F_f)\star\mathscr{R}_V(F_g)) \omega^\beta(K) = \\
 \lambda\int \omega_2^\beta(f,y)\omega_2^\beta(g,y)  \dot\chi_-(y^0+t)h(\vettore{y})  \dvol{y^0} \dvol{^3 \vettore{y}} + O(\lambda^2)
\end{gather*}
where $y=(y^0,\vettore{y})$ and $\dot{\chi}_-$ is given in \eqref{eq:time-average-hamiltonian}. Furthermore, $\omega_2^\beta(f,y):=\langle\omega_2^\beta,f\otimes \delta_y\rangle$ is given in terms of $\delta_y$, the Dirac delta function centered in $y$ and it is a smooth function thanks to the Hadamard property.
Using the exponential decay for large spatial separations of the two-point functions $\omega_2^\beta$ given in \eqref{eq-KMS-free} we have
\[
\lim_{h\to1} \int \omega_2^\beta(f,y)\omega_2^\beta(g,y)  \dot\chi_-(y^0+t)h(\vettore{y})  \dvol{y^0} \dvol{^3 \vettore{y}} 
= \int \omega_2^\beta(f,y)\omega_2^\beta(g,y)  \dot\chi_-(y^0+t)  \dvol{y^0} \dvol{^3 \vettore{y}}  =: \langle F_t ,f\otimes g\rangle.
\]
We shall now study the form of the distribution $F_t$. In particular, 
\eqref{eq-KMS-free} implies that
\begin{gather*}
F_t(x_1,x_2) 
=
\frac{1}{(2\pi)^6}
\int\dvol{\vettore{^3y}}\dvol{y^0}\dot\chi_-(y^0+t)
\\
\cdot
\prod_{j=1}^2
\int\dvol{^3\vettore{k}_j}
\left(b_+(\vettore{k}_j)
\frac{e^{i\frequenza{\vettore{k}_j}(x_j^0-y^0) }}{2\frequenza{\vettore{k}_j}}
+b_-(\vettore{k}_j)
\frac{e^{-i\frequenza{\vettore{k}_j}(x_j^0-y^0) }}{2\frequenza{\vettore{k}_j}}
\right)e^{-i\vettore{k}_j(\vettore{x}_j-\vettore{y})}.
\end{gather*}
The integral in $\dvol{\vettore{y}}$ gives a delta contribution which forces $\vettore{k}_1+\vettore{k}_2=0$.
In the product between the various modes there is $y_0$-independent contribution which remains unaffected by the $\tau$-translation:
\begin{gather}
b_+b_-\left(
\frac{e^{i\frequenza{\vettore{k}}(x_1^0-x_2^0) }}{4\frequenza{\vettore{k}}^2} + \frac{e^{-i\frequenza{\vettore{k}}(x_1^0-x_2^0) }}{4\frequenza{\vettore{k}}^2}
\right)=
\frac{1}{2\frequenza{\vettore{k}}^2}b_+(\vettore{k})b_-(\vettore{k})\cos(\frequenza{\vettore{k}}(x_1^0-x_2^0)).
\end{gather}
The other contributions are proportional to oscillatory phases $\sim e^{i\frequenza{\vettore{k}}\tau}$ and disappear in the limit of large times.
Note that this is guaranteed only in presence of the time average.
Summing up we find
\begin{gather}
w(x_1,x_2):=\lim_{T\to+\infty}\frac{1}{T}\int_{0}^T\dvol{t}\int\dvol{y}\dot\chi_-(y^0+t) \kms{x_1,y}\kms{x_2,y}
\nonumber\\
=\frac{1}{(2\pi)^3}\int\dvol{^3\vettore{k}}
\frac{1}{2\frequenza{\vettore{k}}^2}b_+(\vettore{k})b_-(\vettore{k})\cos(\frequenza{\vettore{k}}(x_1^0-x_2^0))
e^{i\vettore{k}(\vettore{x}_1-\vettore{x}_2)}
\label{eq:2pt-rep}
\end{gather}
where in the last equality we were able to perform the integral of $y^0$ thanks to the form of $\dot\chi_-(y^0+t)$ given in \eqref{eq:time-average-hamiltonian}.
Hence, at first order in perturbation theory,
\begin{gather*}
\lim_{T\to\infty}\lim_{h\to1} \frac{1}{T}\int \dvol{t} 
\left(
\omega^{\beta}(\mathscr{R}_V(F_f)\star\mathscr{R}_V(F_g)\star\alpha_t(K)) -  \omega^{\beta}(\mathscr{R}_V(F_f)\star\mathscr{R}_V(F_g)) \omega^\beta(K) 
\right)
= \\
\lambda w(f,g)+O(\lambda^2).
\end{gather*}
which is in general non-vanishing. We have thus the proof of the proposition.
\end{proof}

{\bf Remark} Notice that the proof of the previous proposition can be directly applied also to the case of a $\lambda \phi^4$ theory over a massive KMS state. 
Actually, in that case, when the $\star_{\omega^\beta}$ product is used to describe the product of the theory, the interacting potential acquires the known contribution called thermal mass. More precisely, from \eqref{eq:gamma} we have that
\[
\gamma_{w_\beta}^{-1} (\lambda\phi^4) = \lambda \left(\phi^4 + M_\beta \phi^2 + C\right).
\]
The distribution $\omega^\beta(x_1,y)\omega^\beta(x_2,y)$ discussed in the proof of Proposition \ref{pr:failure-clustering} can be graphically depicted in the following way
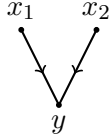
\begin{center}
\begin{tikzpicture}[thick,scale=1] 
\draw[GFleche] (0,0)--(-0.5,1);
\draw[GFleche] (0,0)--(.5,1);
\filldraw (0,0) circle (.7pt) node[below] {$y$};
\filldraw (-.5,1) circle (.7pt) node[above] {$x_1$};
\filldraw (0.5,1) circle (.7pt) node[above] {$x_2$};
\end{tikzpicture}
\end{center}
The essential point in the proof of the previous proposition is in the fact that after having performed the spatial integration over the whole $\vettore{y}-$space a non-vanishing contribution which is constant in $y^0$ remains. Hence, the time average over back-in-time translations $y^0\to y^0-t$ is not vanishing. 
Operating in a similar way, one sees that when more lines are attached to the point $y$, like in
\begin{center}
\begin{tikzpicture}[thick,scale=1] 
\draw[GFleche] (0,0)--(0,1);
\draw[GFleche] (0,0)--(-0.5,1);
\draw[GFleche] (0,0)--(.5,1);
\filldraw (0,0) circle (.7pt) node[below] {$y$};
\filldraw (-.5,1) circle (.7pt) node[above] {$x_1$};
\filldraw (0,1) circle (.7pt) node[above] {$x_2$};
\filldraw (0.5,1) circle (.7pt) node[above] {$x_3$};
\filldraw (1,.5) node[right]{$=\kms{x_1,y}\kms{x_2,y}\kms{x_3,y}$};
\end{tikzpicture}
\end{center}
the corresponding contributions vanish essentially because of the Riemann-Lebesgue lemma.
Coming back to the elements which give non-vanishing contributions we observe that even if the vertex 
$y$ is substituted by two points joined by some internal propagators, due to the momentum conservation, its ergodic mean in the large time limit is again 
non-vanishing.
 
As an example we could consider in a $\lambda \phi^4$ theory with vanishing thermal mass the contribution
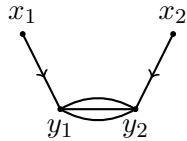
\begin{center}
\begin{tikzpicture}[thick,scale=1] 
\draw[GFleche] (0,0)--(-0.5,1);
\draw[GFleche] (1,0)--(1.5,1);
\draw (0,0)--(1,0);
\draw[bend left] (0,0) edge (1,0);
\draw[bend right] (0,0) edge (1,0);

\filldraw (0,0) circle (.7pt) node[below] {$y_1$};
\filldraw (1,0) circle (.7pt) node[below] {$y_2$};
\filldraw (-0.5,1) circle (.7pt) node[above] {$x_1$};
\filldraw (1.5,1) circle (.7pt) node[above] {$x_2$};
\end{tikzpicture}
\end{center}

The large time limit of $\omega^{V,+}_T$ given in \eqref{eq:average-limit-1} is even more problematic than this. Actually,
in the expansion at higher orders in $K$, there are new contributions which do not converge even if the ergodic mean is considered. We shall see an explicit example in the following proposition

\begin{proposition}\label{pr:infrared-divergences}
Consider a quadratic interaction Lagrangian in the adiabatic limit $(h\to 1)$. Consider the quadratic field $A=\int_M f \phi^2 d^4x$ where $f\in \mathcal{D}(M)$ and $\int dt f(t,\vettore{x}) = 0$ for every $\vettore{x}$.
The contribution
\[
Q^{(n)}_T(A) = \frac{1}{T}\int_0^T \dvol{t_{n+1}} \int_0^{t_{n+1}} \dvol{t_n}  \dots \int_0^{t_2} \dvol{t_1}  \omega^\beta([\alpha_{t_1}(K),\dots, [\alpha_{t_n}(K),\alpha_{t_{n+1}}(A)]]\dots ]) 
\]
to the ergodic mean $\omega^{V,+}_T(A)$ in \eqref{eq:average-limit-1} does not converge for $n\geq 3$ in the sense of perturbation theory for large $T$, if the adiabatic limit is taken in advance.
\end{proposition}
\begin{proof}
We can compute $Q^{(n)}_T(A)$ with a graph expansion as in the proof of Proposition \ref{pr:clustering2}.
Here, we discuss the main points which permits to prove the proposition.
Only connected oriented graphs are present in the graph expansion of $Q^{(n)}_T(A)$.
\begin{gather*}
Q^{(n)}_T(A) =\!\!\!\! \sum_{G\in \mathcal{G}_{n+1}^{o,c}} \!\!\!\!\frac{c_G}{T} \!\!\int\limits_{T S_{n+1}}\!\!\!\!\dvol{t_1}\dots \dvol{t_{n+1}}\! \left. m \!\!\!\!\prod_{l\in E(G)}\!\! D_{s(l)r(l)}\! \left(\alpha_{t_1}(K)\!\otimes\! \dots\! \otimes\! \alpha_{t_n}(K) \!\otimes \alpha_{t_{n+1}}\!(A)\right) \right|_{(\phi_1,\dots \phi_{n+1})=0}\\
=:\frac{1}{T}\int_{T S_{n+1}}\dvol{t_1}\dots \dvol{t_{n+1}}\; F_{n+1}(t_1,\dots , t_{n+1})
\end{gather*}
where $c_G$ is a numerical factor which can also vanish. Furthermore, $D_{ij}$ is defined as in \eqref{eq:differential} and
 $E(G)$ is the set of the lines of $G$, $s(l)$ and $r(l)$ in $\{1,\dots, n+1\}$ indicates the source and the range of the line $l$.
At lowest order in perturbation theory, which is $n$ for $Q^{(n)}_T(A)$ because $K$ introduced in \eqref{eq:time-average-hamiltonian} is at least linear in the perturbation parameter $\lambda$, $\lambda K=\lambda H+O(\lambda^2)$, only oriented connected graphs with $n+1$ vertices and $n+1$ lines contribute to $Q^{(n)}_T(A)$. Hence, every vertex is either the source or the range of exactly two lines.
Actually, at order $n$ in perturbation theory,
\[
F_{n+1}^{(n)}(t_1,\dots, t_{n+1}) = \frac{1}{2}m(D_{12}^2-D_{21}^2)\left( \alpha_{t_1}(H)\otimes   B_{n}(t_2,\dots, t_{n+1})    \right)
\]
where for $n\geq 2$
\[
B_{n}(t_1,\dots, t_{n}):=m(D_{12}-D_{21})\left( \alpha_{t_2}(H) \otimes B_{n-1}(t_2,\dots, t_{n})\right), \qquad   B_{1}(t):=\alpha_{t}(A)\;.
\]
Notice that, for every $n$, $B_n$ are quadratic fields.
Furthermore, notice that
\[
D_{12}-D_{21} = 
\left\langle i\Delta,\frac{\delta^2}{\delta \varphi_1\delta \varphi_2} \right\rangle
\]
where $\Delta(x,y) =-i \omega_2^\beta(x,y)+i\omega_2^\beta(y,x)$ is the causal propagator.  
Since $\Delta(x,y)$ vanishes for points $(x,y)$ with spacelike separation, and since $\frac{\delta H_{h,\chi}}{\delta \phi(x)}$ with $H$ defined in \eqref{eq:time-average-hamiltonian} is timelike compact uniformly in $h$ and $\frac{\delta A}{\delta \phi(x)}$ is compact, we have that for every $n$ and at fixed $t_1,\dots, t_{n}$, $B_{n}(t_1,\dots, t_{n})$ is of compact support. Hence, $B_{n}\in \mathscr{F}_{\mu c}$ even if the adiabatic limit $h\to 1$ is considered. 
For a similar reason the adiabatic limit can be easily taken also in $F_{n+1}^{(n)}$. Actually, $\omega_2^\beta(x,y)\omega_2^\beta(x,z)-\omega_2^\beta(y,x)\omega_2^\beta(z,x)$ which appears in the expansion of $(D_{12}^2-D_{21}^2)$, vanishes if both $x-y$ and $x-z$ are space like vectors.
Notice that, once the spatial Fourier transform is taken, using the form of the two-point function \eqref{eq-KMS-free}, 
$F_{n+1}^{(n)}(t_1,\dots, t_{n+1})$ can be computed directly. 

Here, in order to analyze the integral of $F_{n+1}^{(n)}$ over the simplex $T S_{n+1}$, we further decompose it as a sum over the copies of disjoint subsets of $\{1,\dots, n+1\}$
\begin{gather*}
F_{n+1}^{(n)}(t_1,\dots, t_{n+1}) = \\
\sum_{\substack{I,J\subset \{1,\dots, n+1\}\\ I\cap J = \emptyset}} \int \dvol{^3\vettore{p}} 
\left(e^{i2\frequenza{\vettore{p}}(\sum_{i\in I} t_i-\sum_{j\in J}t_j)}\widehat\Phi_{I,J,+}(\vettore{p}) +
e^{-i2\frequenza{2\vettore{p}}(\sum_{i\in I} t_i-\sum_{j\in J}t_j)}\widehat\Phi_{I,J,-}(\vettore{p})
\right)
\end{gather*}
where $\widehat{\Phi}_{I,J,\pm}(\vettore{p})$ are suitable functions which are rapidly decreasing for large spatial momentum $\vettore{p}$. 

We observe that, the largest contribution in $Q^{(n)}_T(A)$ is obtained when $|I|$ and $|J|$ are small. By direct inspection we notice that when either $I$ and $J$ are empty sets, both $\Phi_{I,J,\pm}$ vanish. When both $I$ and $J$ contain only one element, due to the form of $A$, the only non-vanishing contribution  in the sum is when $I=\{t_1\}$ and $J=\{t_{n+1}\}$ or when 
$I=\{t_{n+1}\}$ and $J=\{t_1\}$. Hence, 
the most divergent contribution for large times $T$ to $Q^{(n)}_T(A)$ is given by 
\begin{gather*}
F_{n+1}^{(n)}(t_1,\dots, t_{n+1}) = \\
C_n  \int \dvol{^3\vettore{p}}  (b_++b_-)   
\left(
\frac{e^{i2\frequenza{\vettore{p}}(t_1-t_{n+1})}}{\frequenza{\vettore{p}}^{n+1}} \widehat\Phi_+(\vettore{p})
+ (-1)^n\frac{e^{-i2\frequenza{\vettore{p}}(t_1-t_{n+1})}}{\frequenza{\vettore{p}}^{n+1}}\widehat\Phi_-(\vettore{p})
\right)+R
\end{gather*}
where $C_n$ is a numerical factor and
$\widehat\Phi_+(\vettore{p})=\widehat{\dot\chi_-}(-\frequenza{\vettore{p}}) \widehat{f}(\frequenza{\vettore{p}},\vettore{p})$
$\widehat\Phi_-(\vettore{p})=\widehat{\dot\chi_-}(\frequenza{\vettore{p}}) \widehat{f}(-\frequenza{\vettore{p}},\vettore{p})$.

The integral over the  simplex $T S_{n+1}$ can be computed before the integration over $\vettore{p}$. It gives an oscillating function whose amplitude grows as $T^n$, times $e^{\pm i \frequenza{\vettore{p}} T}$.  The integration over $\vettore{p}$ gives again an oscillating function whose amplitude decays as $1/T^{3/2}$. Combining these two contributions we have that
the amplitude of $Q^{(n)}_T(A)$ grows as $T^{n-3/2-1}$, hence for $n\geq 3$ $Q^{(n)}_T(A)$ does not converge for large $T$.
\end{proof}

These kind of infrared divergences can be traced back to the difficulties present in the analysis of the existence of adiabatic limit \cite{Altherr, Landsman}.
In the literature it is claimed that, they can be tamed by partial resummations of the perturbative series. However, these kind of resummations are beyond perturbation theory.

\section{A non-equilibrium steady state for the free field theory}

In the previous section we have seen in Proposition \ref{pr:failure-clustering} that the clustering condition does not hold in the adiabatic limit. 
Hence, we expect that, when the ergodic mean of a state converges to another state, these new state is a non-equilibrium steady state \cite{Ru00}. 
We have however seen that some infrared divergences are present in the ergodic mean of $\omega^{\beta}\circ \alpha_\tau^V$, see e.g. Proposition \ref{pr:infrared-divergences}. 
Contrary to this situation Fredenhagen and Lindner have shown in \cite{Lindner, FredenhagenLindner} that a KMS state $\omega^{\beta, V}$ for interacting theory can be constructed also in the adiabatic limit.
Hence, in order to construct an example of a non-equilibrium steady state, we revert the point of view and we analyze the ergodic mean of $\omega^{\beta,V}$ with respect to the free time evolution $\translation{\tau}$
\begin{gather}\label{eq:average-limit-2}
\omega^{+}(A):=\lim_{T\to\infty}\lim_{h\to 1}\frac{1}{T}\int_{0}^T\omega^{\beta,V}({\translation{\tau}(A)})\dvol{\tau}
\end{gather}
which is seen as a state (defined as a formal power series) for the unperturbed theory.

\begin{theorem}\label{th:existence}
The functional $\omega^{+}$ defined in the sense of formal power series in \eqref{eq:average-limit-2}, is a state for the free algebra $\mathscr{F}$. Furthermore, $\omega^{+}$ is invariant under the free evolution $\alpha_t$.
\end{theorem}
\begin{proof}
$\omega^{\beta,V}$ defined in \eqref{eq:interacting-KMS} is a linear functional over the free algebra given in the sense of formal power series. Furthermore, it is normalized by construction. It is also positive again in the sense of formal power series. It is thus a state over the free algebra.
$\omega^{\beta,V}\circ\alpha_t$ is a state because it is the composition of a state with an automorphism of the algebra. 
The properties of being  positive, normalized and linear are  preserved by the ergodic mean of functionals.  

In order to prove that for every $A$ the limit $T\to\infty$ exists let us start analyzing 
\begin{equation}\label{eq:da-stimare}
\omega^{\beta,V}(A) =  \omega^{\beta}(A) +  \sum_{n\geq 1} (-1)^n  \int_0^{\beta} \dvol{u_n} \int_0^{u_n} \dvol{u_{n-1}} \dots   \int_0^{u_{2}} \dvol{u_1} \; \omega^{\beta,c}(A \otimes \alpha_{iu_1}(K)\otimes \dots \otimes \alpha_{iu_n}(K))
\end{equation}
furthermore in \cite{FredenhagenLindner} (see also appendix C in \cite{DHP}) it is shown that the state in the adiabatic limit can be obtained in the following way
\begin{equation}\label{eq:da-stimare-2}
\omega^{\beta,V}(A) =  \omega^{\beta}(A) +  \sum_{n\geq 1} \int_{\beta S_n} \dvol{u_n} \dots \dvol{u_1}  \int_{\mathbb{R}^3} \dvol{^3{\vettore{x}}_1}\dots \int_{\mathbb{R}^3} \dvol{^3{\vettore{x}}_n}  \; \omega^{\beta,c}(A \otimes \alpha_{iu_1,{\vettore{x}}_1}(R)\otimes \dots \otimes \alpha_{iu_n,{\vettore{x}}_n}(R))
\end{equation}
where $R := -\mathscr{R}_{V_{\chi,1}}(H(\dot\chi^- \delta_0))$, $\delta_0$ is the Dirac delta function centered in the origin of $\mathbb{R}^3$ and $H$ is given in \eqref{eq:time-average-hamiltonian}.
We are thus interested in analyzing  
\begin{gather}
\omega^{\beta,V}(\alpha_T(A)) =  \omega^{\beta}(\alpha_T(A)) 
\notag
\\
\label{eq:mean-ergodic}
+  \sum_{n\geq 1} \int_{\beta S_n} \dvol{u_n} \dots \dvol{u_1}  \int_{\mathbb{R}^3} \dvol{^3{\vettore{x}}_1}\dots \int_{\mathbb{R}^3} \dvol{^3{\vettore{x}}_n}  \; \omega^{\beta,c}(\alpha_T(A) \otimes \alpha_{iu_1,{\vettore{x}}_1}(R)\otimes \dots \otimes \alpha_{iu_n,{\vettore{x}}_n}(R))
\end{gather}
for large values of $T$.
Hence, let us study how the following function, related to the integrand of \eqref{eq:mean-ergodic}, depends on $T$
\[
F_n(u_0-iT,\vettore{x}_0; u_1,\vettore{x}_1;\dots  ;u_n,\vettore{x}_n)
:=
\omega^{\beta,c}(\alpha_{iu_0+T}(A) \otimes \alpha_{iu_1,{\vettore{x}}_1}(A_1)\otimes \dots \otimes \alpha_{iu_n,{\vettore{x}}_n}(A_n)).
\]
In the first part of the following analysis we follow the proof of Theorem 4 in \cite{FredenhagenLindner}. For completeness, we shall just recall the main steps. 
Due to the integration domain in \eqref{eq:da-stimare}
we are interested in the case 
\begin{equation}\label{eq:order}
0=u_0 < u_1 <\dots < u_n < \beta 
\end{equation}
and $\vettore{x}_0=0$.

Furthermore, without loosing generality, 
we might restrict our attention to the case where ${u_{i+1}}-{u_i}\leq \frac{\beta}{2}$ for every $i$. Actually, if for some $m<n$, $u_{m+1}-u_m > \beta/2$, \eqref{eq:order} implies that $u_j-u_0 < \beta/2$ for every $j\leq m$ and  $\beta-u_j<\beta/2$ for $j\geq m+1$. By the KMS condition we have that
\begin{gather*}
\omega^{\beta,c}(\alpha_{iu_0+T}(A) \otimes \alpha_{iu_1,{\vettore{x}}_1}(A_1)\otimes \dots \otimes \alpha_{iu_n,{\vettore{x}}_n}(A_n))\\
=\omega^{\beta,c}(\alpha_{iu_{m+1},\vettore{x}_{m+1} } (A_{m+1}) \otimes \dots \otimes \alpha_{iu_n,{\vettore{x}}_n}(A_n)
\otimes \alpha_{i\beta+iu_0+T}(A) \otimes \dots \otimes \alpha_{i\beta+iu_m,{\vettore{x}}_m}(A_m) )\\
=:F_n'(u_{m+1},\vettore{x}_{m+1}; \dots ;u_n,\vettore{x}_n;  u_0+\beta-iT,\vettore{x}_0; u_1+\beta,\vettore{x}_1;\dots  ;u_m+\beta,\vettore{x}_m)
\end{gather*}
and now, the argument of the new function $F'_n$ have the desired property, actually $u_{m+1} < \dots < u_n < \beta+u_0 < \dots <\beta+u_m $ and $\beta+u_m-u_{m+1} < \beta/2$, hence, indicating by $\tilde{u}_i$ the new arguments of $F_n'$ we have that $\tilde{u}_j-\tilde{u}_i <\beta/2$ for every $i<j$.

From the definition of connected functions, it descends that $F_n$ can be decomposed as a sum over $\mathcal{G}^c_{n+1}$ the set of connected graphs with $n+1$ vertices
\begin{gather*}
F_n(u_0-iT,\vettore{x}_0; u_1,\vettore{x}_1;\dots  ;u_n,\vettore{x}_n) = \\ 
\sum_{G\in \mathcal{G}^c_{n+1}} \prod_{i<j} \left.\left( \frac{D_{ij}^{l_{ij}}}{l_{ij}!} \right)
(
\alpha_{iu_0+T}(A) \otimes \alpha_{iu_1,{\vettore{x}}_1}(A_1)\otimes \dots \otimes \alpha_{iu_n,{\vettore{x}}_n}(A_n)
)
\right|_{(\phi_0,\dots,\phi_n)=0}\\
=:
\sum_{G\in \mathcal{G}^c_{n+1}} \prod_{i<j} \frac{1}{\text{Symm}(G)} F_{n,G}(u_0-iT,\vettore{x}_0; u_1,\vettore{x}_1;\dots  ;u_n,\vettore{x}_n)
\end{gather*}
where $l_{ij}$ denotes the number of lines joining the vertices $i$ and $j$ in $G$ and $\text{Symm}(G)$ is a suitable numerical factor. 
In the proof of Theorem 4 in \cite{FredenhagenLindner}, $F_{n,G}$ is then expanded as   
\begin{gather}\label{eq-positive-negative}
F_{n,G}(u_0,{\vettore{x}}_0; u_1,\vettore{x}_1;\dots  ;u_n,\vettore{x}_n) =  \int \dvol{P} \prod_{l\in E(G)}
\frac{e^{i\vettore{p}_l(\vettore{x}_{s(l)}-\vettore{x}_{r(l)})} (\lambda_{+}(p_l)+\lambda_{-}(p_l))}{2\omega_l(1-e^{-\beta \omega_l})}
\hat{\Psi}(-P,P)
\end{gather}
where $E(G)$ is the set of lines of the graph $G$, $s(l)$  and $r(l)$ are respectively the indexes of the source and the range of the points joined by the line $l$.
\[
\Psi(X,Y) = 
\left.
\prod_{l\in E(G)} \frac{\delta^2}{\delta \phi_{s(l)}(x_l)\delta \phi_{r(l)}(y_l)} (A_0\otimes \dots \otimes A_n)
\right|_{(\phi_0,\dots,\phi_n)=0}
\]
so that $X$ and $Y$ are for $(x_1,\dots ,x_k)$ and $(y_1,\dots ,y_k)$ and $k$ indicates the number of lines in $E(G)$. $\hat{\Psi}(-P,P)$ is the Fourier transform of $\Psi(X,Y)$ and hence $P=(p_1,\dots ,p_k)$. 
Moreover, the positive and negative frequency part in $D_{ij}$ are indicated by
\begin{equation}\label{eq:lambdas}
\lambda_+(p_l) = e^{-\omega_l (u_{r(l)}-u_{s(l)})}\delta(p_l^0-\omega_l), \qquad
\lambda_-(p_l) = e^{\omega_l (u_{r(l)}-u_{s(l)}-\beta)}\delta(p_l^0+\omega_l), \qquad
\end{equation}
with $\omega_l=\sqrt{\vettore{p}_l^2+m^2}$. The sum over the positive and negative frequency parts present in \eqref{eq-positive-negative} can be further expanded as follows
\begin{gather}
F_{n,G}(u_0,{\vettore{x}}_0; u_1,\vettore{x}_1;\dots  ;u_n,\vettore{x}_n) =  \sum_{\{E_+,E_-\}\in P_2(E(G))}\int \dvol{P} 
\prod_{l_+\in E_+} \frac{e^{i\vettore{p}_{l_+}(\vettore{x}_{s(l_+)}-\vettore{x}_{r(l_+)})} \lambda_{+}(p_{l_+})}{2\omega_{l_+}(1-e^{-\beta \omega_{l_+}})}
\notag
\\
\label{eq:decomp}
\cdot
\prod_{l_-\in E_-} \frac{e^{i\vettore{p}_{l_-}(\vettore{x}_{s(l_-)}-\vettore{x}_{r(l_-)})} \lambda_{-}(p_{l_-})}{2\omega_{l_-}(1-e^{-\beta \omega_{l_-}})}
\hat{\Psi}(-P,P)
\end{gather}
where $P_2(E(G))$ is the set of all possible partitions of $E(G)$ into two disjoint subsets $E(G)=E_+ \cup E_-$ (which can be empty) one for the positive and one for the negative part.
Notice that for every $l\in E(G)$, $u_{r(l)}-u_{s(l)}\leq \beta/2$, hence, 
the argument of the exponential in $\lambda_-(p_{l_-})$ is always bigger then $-\beta/2 \omega_{l_-}$, that is
\[
e^{\omega_{l_-}(u_{r(l)}-u_{s(l)}-\beta)} \leq e^{-\omega_{l_-} \beta/2}.
\]
The function $\hat{\Psi}(-P,P)$ is the Fourier transform of a compactly supported distribution, hence, it is an entire analytic function which grows at most polynomially in every direction $(-P,P)$, hence,
\[
\hat{\Phi}(P):=\prod_{l_-\in E_-} e^{\omega_{l_-}(u_{r(l_-)}-u_{s(l_-)}-\beta)}\hat{\Psi}(-P,P)
\]
is rapidly decreasing in every direction containing negative frequencies (at least one $p_{l_-}\in P$ has $p^0_{l_-}<0$). 
Furthermore, by Proposition 8 in \cite{FredenhagenLindner}, $\hat{\Psi}(-P,P)$ is of rapid decrease in the directions $P$ contained in $(V^+)^k$ the $k-$fold product of the forward light cone.   
Finally, since, the $\delta$ functions in \eqref{eq:lambdas} forces 
$p_{j}, \forall j\in E_+$ and $p_{k}, \forall k\in E_-$
to be respectively on the positive and negative mass shell, 
we have that 
\[
\hat{\Phi}_m(\vettore{P}):= \left. \prod_{l_-\in E_-} e^{\omega_{l_-}(u_{r(l_-)}-u_{s(l_-)}-\beta)}\hat{\Psi}(-P,P)  
\right|_{p_{j}^0=\pm \omega_{j}, \forall j \in E_{\pm}}
\]
is of rapid decrease in every spatial momenta $\vettore{P}$. In particular, this implies that, every integrand in \eqref{eq:decomp} is absolutely integrable.
Furthermore, the spatial integrals in \eqref{eq:da-stimare-2} can be performed to obtain
\begin{gather*}
\int_{\mathbb{R}^3} \dvol{^3{\vettore{x}}_1}\dots \int_{\mathbb{R}^3} \dvol{^3{\vettore{x}}_n} \;
F_{n,G}(u_0-iT,0; u_1,\vettore{x}_1;\dots  ;u_n,\vettore{x}_n)\\
=
\left.\sum_{\{E_+,E_-\}\in P_2(E(G))}\int \dvol{\vettore{P}} 
\prod_{l_+\in E_+} \frac{e^{-\omega_{l_+}(u_{r(l_+)}-u_{s(l_+)})}}{2 \omega_{l_+}(1-e^{-\beta \omega_{l_+}})}
\prod_{l_-\in E_-} \frac{e^{\omega_{l_-}(u_{r(l_-)}-u_{s(l_-)}-\beta)}}{2\omega_{l_-}(1-e^{-\beta \omega_{l_-}})}
\hat{\Psi}(-P,P) 
\right|_{\substack{p_{j}^0=\pm\omega_{j}, \\ \forall j\in E_\pm}}
\\
\cdot
\prod_{i\in \{1,\dots, n\}} \delta\left(\sum_{\substack{l\in E(G) \\ s(l)=i}} \vettore{p}_l-\sum_{\substack{l\in E(G) \\ r(l)=i}} \vettore{p}_l\right)
\prod_{\substack{e_+\in E_+ \\ s(e_+)=0}} e^{-iT \omega_{e_+}}\prod_{\substack{e_-\in E_- \\ s(e_-)=0}} e^{iT \omega_{e_-}} 
\end{gather*}
the delta functions over the linear combinations of various $p_l$ is the spatial momentum conservation at all but one entry of the tensor product. 
Since $G$ is a connected graph with $n+1$ vertices, the number of lines, $k$ in $G$, is always larger or equal to $n$. 
The integration over $\vettore{P}$ is thus an integration over $k-$spatial momenta. The products of $n$ delta functions over some linear combination of various $\vettore{p}_l$ is thus a well defined distribution provided the $n$ linear combinations $\vettore{c}_i:=\sum_{l\in E(G), s(l)=i} \vettore{p}_l-\sum_{l\in E(G),r(l)=i} \vettore{p}_l$ with $i\in\{1,\dots, n\}$ are independent in a neighborhood of the support of the delta functions.  The latter condition is again ensured by the fact that the graph $G$ is connected. 
This fact can be proved checking the maximality of the rank of the $3n\times 3k$ matrix formally defined as $J_{n,k}:=\{\frac{\partial \vettore{c}_i}{\partial \vettore{p}_l}\}_{i\in \{1,\dots, n\}; l\in \{1,\dots, k\}}$. For graphs of $n+1$ points with $n$ lines this can be proven by induction on the number of points $n$. Let $G_n$  be such a graph, then there is always at least one point labeled by $i\neq 0$ which is reached by only one line. If we remove that point and that line from the graph $G_n$ we obtain another connected graph $G_{n-1}$ of $n$ points with $n-1$ lines.   The corresponding matrices $J_{n-1,n-1}$ and $J_{n,n}$ are such that $\det J_{n,n}=\pm J_{n-1,n-1}$. Since $\det J_{1,1}=\pm 1$ this finishes the proof in that case.
Finally, if $k$, the number of lines in the connected graph $G$ with $n+1$ points, is larger than $n$, it is always possible to find a connected subgraph $G'_n$ which has exactly $n$ lines. In this case $J_{n,n}$ is just a submatrix of $J_{n,k}$ and the maximality of the rank of $J_{n,k}$ is ensured by that of $J_{n,n}$.

Furthermore, the ergodic mean of the previous expression gives
\begin{gather*}
\lim_{T\to\infty}\frac{1}{T} \int_0^T d\tau
\int_{\mathbb{R}^3} \dvol{^3{\vettore{x}}_1}\dots \int_{\mathbb{R}^3} \dvol{^3{\vettore{x}}_n} \;
 F_{n,G}(u_0-i\tau,0; u_1,\vettore{x}_1;\dots  ;u_n,\vettore{x}_n)\\
=
\lim_{T\to\infty}\!\!
\left.\sum_{\{E_+,E_-\}\in P_2(E(G))}\int \dvol{\vettore{P}} 
\!\!\!\prod_{l_+\in E_+} \frac{e^{-\omega_{l_+}(u_{r(l_+)}-u_{s(l_+)})}}{2 \omega_{l_+}(1-e^{-\beta \omega_{l_+}})}
\prod_{l_-\in E_-} \!\!\frac{e^{\omega_{l_-}(u_{r(l_-)}-u_{s(l_-)}-\beta)}}{2\omega_{l_-}(1-e^{-\beta \omega_{l_-}})}
\hat{\Psi}(-P,P) 
\right|_{\substack{p_{j}^0=\pm\omega_{j}, \\ \forall j\in E_\pm}}
\\
\cdot
\prod_{i\in \{1,\dots, n\}} \delta\left(\sum_{\substack{l\in E(G) \\ s(l)=i}} \vettore{p}_l-\sum_{\substack{l\in E(G) \\ r(l)=i}} \vettore{p}_l\right)\frac{1-e^{iT\left(
\sum_{\substack{e_-\in E_- \\ s(e_-)=0}} \omega_{e_-} 
-\sum_{\substack{e_+\in E_+ \\ s(e_+)=0}} \omega_{e_+}
\right)}}{iT \left(
\sum_{\substack{e_-\in E_- \\ s(e_-)=0}} \omega_{e_-} 
-\sum_{\substack{e_+\in E_+ \\ s(e_+)=0}} \omega_{e_+}
\right)}
\end{gather*}
Notice that, $(1-e^{i\alpha T})/(\alpha T)$ is bounded by a constant uniformly in $\alpha$ and $T$. Hence, after applying the delta functions, the limit $T\to\infty$ can be taken before the integral over the remaining momenta. This limit vanishes unless  
$
\sum_{\substack{e_-\in E_-, s(e_-)=0}} \omega_{e_-} 
-\sum_{\substack{e_+\in E_+, s(e_+)=0}} \omega_{e_+}
 = 0 
$
on a set of non-zero measure over the remaining momenta, if any. In the latter case it furnishes a finite result. 
\\
Due to absolute convergence, the ergodic mean and the corresponding limit for $T\to \infty$ of \eqref{eq:mean-ergodic} can be taken before the integral over $(u_1,\dots, u_n)\in \beta S_n$. Hence we have the result.
\end{proof}

The expression obtained at the last step of the previous proof and the fact that there are cases where it is non-vanishing shows that in general $\omega^+$ is different from $\omega^\beta$.
This is another indirect evidence of the failure of the clustering condition under the adiabatic limit established in Proposition \ref{pr:failure-clustering}.
Actually,
the failure of the clustering property under the adiabatic limit, shows that the state $\omega^{+} \neq \omega^{\beta}$. We shall now see that 
$\omega^{+}$ does not satisfy the KMS condition with respect to $\alpha_t$.
\begin{theorem}
$\omega^{+}$ does not satisfy the KMS condition with respect to $\alpha_t$.
\end{theorem}
\begin{proof}
The state $\omega^{\beta}$ is a KMS state with respect to $\alpha_t$. Hence let us test the KMS condition on the difference
\[
w(A):=\omega^{\beta}(A)-\omega^{+}(A)=\omega^{\beta}(A)-\lim_{T\to\infty}\frac{1}{T}\int_0^T \dvol{\tau} \; \omega^{\beta,V}(\alpha_\tau(A)) 
\]
recalling \eqref{eq:interacting-KMS} and \eqref{eq:pert-dyn} we notice that the contribution of order $0$ in $\lambda$ in the previous expression vanishes, furthermore,  operating as in the proof of Theorem \ref{th:stability-compact} we have that 
\[
w(A) = \lim_{T\to\infty}\frac{1}{T}\int_0^T \dvol{\tau} \; \int_0^\beta \dvol{u} \; \omega^{c,\beta}(A\otimes \alpha_{iu-\tau}(K)) +O(\lambda^2)
\]
Hence, we have that $w(A)$ is related to the failure of averaged clustering condition given in Proposition \ref{pr:failure-clustering}.  
Let us now specialize the analysis of $w(A)$ for $A=\mathscr{R}_V(F_f)\star\mathscr{R}_V(F_g)$ and $V$ is a quadratic potential.
Operating as in the proof of Proposition \ref{pr:failure-clustering}, we have that  
\[
w(\mathscr{R}_V(F_f)\star\mathscr{R}_V(F_g)) = w(f,g)   + O(\lambda^2)
\]
where $w(f,g)$ has the form already given in \eqref{eq:2pt-rep}. The translation $\alpha_{iu}$ present in $w(A)$ has no effect at first order in perturbation theory.
The function $t\mapsto w(f,g_t)$ does not satisfy the KMS property. Actually, the integral kernel of $w$ has the form
\[
w(x_1,x_2)=\beta\int\dvol{^3\vettore{k}}
\frac{1}{4\frequenza{\vettore{k}}^2}\frac{\cos(\frequenza{\vettore{k}}(x_1^0-x_2^0))}{(\cosh(\beta \frequenza{\vettore{k}})-1 )}
e^{i\vettore{k}(\vettore{x}_1-\vettore{x}_2)}
\]
by direct inspection we see that the map $t\mapsto w(f,g_t)$ can be analytically continued to the strip $\Im(t)\in [0,\beta]$, however, 
\begin{gather*}
w(x_1,x_2+i\beta e)-w(x_2,x_1)\\=\beta\int\dvol{^3\vettore{k}}
\frac{1}{4\frequenza{\vettore{k}}^2}
\left(\cos(\frequenza{\vettore{k}}(x_1^0-x_2^0)) +i\sin(\frequenza{\vettore{k}}(x_1^0-x_2^0)) \tanh\left(\frac{\beta \frequenza{\vettore{k}}}{2}\right) \right)
e^{i\vettore{k}(\vettore{x}_1-\vettore{x}_2)}
\end{gather*}
where the four vector $e=(1,0,0,0)$.
Hence, the KMS condition for $\omega^{+}$ does not hold. 
\end{proof}

Notice that, when $V$ describes a perturbation of the mass square $m^2$ to  $m^2+\delta m^2$, namely when $V$ is quadratic in the field and no field derivatives are present, the state $\omega^+$ can be constructed exactly. In this case, the mode decomposition of the two-point function associated with $\omega^{\beta,V}$, looks like \eqref{eq-KMS-free}, furthermore, we have that the action of $\translation{\tau}$ and the corresponding time averaged used in the definition of $\omega^+$ in \eqref{eq:average-limit-2}, transforms the modes $e^{i\omega_{\vettore{k}} t}$ of the $m^2+\delta m^2$ theory in the modes of the corresponding free theory of mass $m$. This procedure does not alter the form of the Bose factor $b_+$, which is the Bose factor at square mass $m^2+\delta m^2$. It is thus clear that $\omega^+$ cannot be a KMS state.

In the general case, 
even if the state $\omega^+$ we have obtained is invariant under time translations it is not a KMS state, hence, it can be seen as a 
NESS for a massive scalar field theory. 
To analyze the thermodynamical properties of NESS in the context of $C^*-$dynamical system,  the notion of entropy production was introduced by 
Ojima et al. \cite{Oj0,Oj1,Oj2} and by Jak{\v s}i\'c and Pillet \cite{JP01,JP02}. A direct generalization of this concept to the present case seems not to be possible due to the presence of infrared divergences. Despite this fact, we expect that when spatial densities are considered 
some of the known results can be recovered. In any case, further investigations are needed. 

\section{Conclusion and open questions}

In this paper, we have tested some of the stability notions known for $C^*-$dynamical systems in the case of interacting quantum scalar field theories treated by means of perturbation theory.
In particular, whenever  $V_{\chi,h}$ is a compactly supported interaction Lagrangian,
we have shown that the extremal free KMS state $\omega^\beta$ for a massive Klein Gordon field evolved with the perturbed dynamics $\alpha^V$, in the limit $t\to \infty$ 
tends to the interacting KMS state $\omega^{\beta,V}$ recently constructed in \cite{FredenhagenLindner}, i.e. the return to equilibrium $\lim_{t\to \infty}\omega^\beta\circ \alpha_t^V= \omega^{\beta,V}$ holds.
Conversely, if the adiabatic limit is considered this is not anymore true. Furthermore, an ergodic (temporal) mean applied to $\omega^\beta\circ \alpha^V$, is not sufficient to repair the return to equilibrium. 
Actually, in this case, the limit of the ergodic mean is not even convergent in the sense of perturbation theory. 
For this reason we have reverted the point of view, studying $\omega^{\beta,V}\circ \alpha_t$. Even if constructed perturbatively, this can be seen as a state for the free theory. It holds that the ergodic mean of this family of states converges to a non-equilibrium steady state $\omega^+$ \cite{Ru00}. 

Some of the thermodynamical properties of the states discussed in the present paper can be addressed directly.
In particular, the state  $\omega^{\beta,V}$ seen as a state over the free algebra, is a perturbation of $\omega^{\beta}$. We might thus evaluate the relative entropy in the two states extending the Araki definition, see also the discussion present in \cite{BR}, to this case: 
\[
S(\omega^\beta|\omega^{\beta,V}) := -\beta\omega^\beta(K) -\log \left(\omega^\beta(U(i\beta))\right).
\]
Notice that, although in the present case  $S(\omega^\beta|\omega^{\beta,V})$ is defined in the sense of perturbation theory, 
 some nice properties of the relative entropy which holds in the context of $C^*-$algebras are still valid.
In particular, both $S(\omega^\beta|\omega^{\beta,V})\leq 0$ and $S(\omega^{\beta,V}|\omega^\beta)\leq 0$.
However, while $S(\omega^\beta|\omega^{\beta,V})$ is finite when $V$ is of spatial compact support, it diverges under the adiabatic limit. 
In \cite{Lindner} it has been argued that spatial densities of similar quantities are finite and preserve some of the thermodynamical 
properties.
It is nevertheless an open question if an entropy production density with similar properties as those studied in \cite{JP01,JP02} can be defined 
for $\omega^+$.

Another open question, which remains to be answered, is about the infrared divergences shown in the analysis of the limit $\omega^\beta\circ \alpha_t^V$ for large $t$. Since $U_V$ is formally unitary, it should be possible to treat some of them with certain partial resummation methods, see e.g. \cite{Landsman, LeBellac}. 
Finally, a comparison of the asymptotic behavior of the interacting thermal two-point function discussed by Bros and Buchholz in \cite{BB02} with the results presented in section \ref{sec:inst} could throw some light on the existence of the limit of the ergodic mean of $\omega^\beta\circ \alpha_t^V$ given in \eqref{eq:average-limit-1} in the non perturbative regime.

\subsection*{Acknowledgments}
We are grateful to K. Fredenhagen for helpful discussions on the subject, to M. Fr\"ob for useful comments on an earlier version of this paper and to T.-P. Hack for suggesting us a proof of Proposition A.1.

\appendix

\section{Some technical propositions}
\begin{proposition}\label{pr-decay-omega-beta}
Consider the two-point function of the extremal free massive KMS state at inverse temperature $\beta$ on the Minkowski spacetime given in \eqref{eq-KMS-free}. If $y-x$ is a timelike future pointing vector,  
\begin{equation}\label{eq:est-2pt}
\left|D^{(\alpha)}\omega_2^\beta(x; t_y+t, \vettore{y})\right|\leq \frac{C_\alpha}{t^{3/2}}  , \qquad t > 1,  
\end{equation}
where $\alpha\in \mathbb{N}^8$ is a multiindex and $D^{(\alpha)}$ indicates the composition of partial derivatives of order $\alpha_i$ along the $i-$th direction in $M^2$  for various $i\in\{1,\dots,8\}$.
\end{proposition}
\begin{proof} A proof of this proposition can be written following Appendix A in \cite{BB02}. Here for completeness we sketch its main steps.
We know that $\omega_2^\beta$ is an Hadamard two-point function, hence, $\omega_2^\beta-\omega_2^{\infty}$ is smooth. 
If $y-x$ is a timelike future pointing vector, the points $(t_y+t, \vettore{y})$ and $y$ are connected by a timelike geodesic for every $t$, and thus 
$(x; t_y+t, \vettore{y})$ is contained neither in the singular support of $\omega_2^\beta$ nor of $\omega_2^\infty$, hence the two-point functions 
$\omega_2^\beta$ and $\omega_2^\infty$ are both described by a smooth function in a neighborhood of $(x; t_y+t, \vettore{y})$. We recall that the vacuum massive two-point function (see e.g. \cite{BDF}) takes the form
\[
\omega_2^{\infty}(x; t_y+t, \vettore{y})=  4\pi \frac{m}{i \sqrt{(t+t_y-t_x)^2-|\vettore{x}-\vettore{y}|^2}}K_1(i m \sqrt{(t+t_y-t_x)^2-|\vettore{x}-\vettore{y}|^2}) 
\] 
where $K_1$ is the modified Bessel function of second kind and of index $1$. 
Hence, using the asymptotic form of the modified Bessel function, see e.g. 8.451 6. in \cite{Grad}, for large values of $t$, and at fixed $x$ and $y$,  $|\omega_2^{\infty}(x; t_y+t, \vettore{y})|\leq C/t^{3/2}$. 
Let us consider
\[
t^{3/2} (\omega_2^\beta-\omega_2^\infty)(0; t, \vettore{x}) = c t^{3/2} \sum_{\sigma\in\{+1,-1\}}\int_{m}^\infty   \dvol{E} (\sqrt{{E^2-m^2}}) \frac{\sin(\sqrt{E^2-m^2} |\vettore{x}|)}{\sqrt{E^2-m^2} |\vettore{x}|}  
\frac{1}{e^{\beta E}-1} e^{i \sigma E t} 
\]
for $t> |\vettore{x}|$ and for a suitable constant $c$. Performing a change of integration variable $w=(E-m) t$ we obtain 
\[
t^{3/2} (\omega_2^\beta-\omega_2^\infty)(0; t, \vettore{x}) = 
\sum_{\sigma\in\{+1,-1\}}\int_{0}^\infty   \dvol{w} \sqrt{w} 
\; b_{\sigma}\left(\frac{w}{t}\right) e^{i\sigma w }
\]
where $b_\sigma$ is a suitable bounded function which decays rapidly for large values of its argument. The integral in the right hand side of the previous expression can be proven to be bounded uniformly in $t$ operating in the following way.
First of all let us isolate the term $b_{\sigma}(0)$ dividing the integral in two parts 
\begin{gather*}
\int_{0}^\infty   \dvol{w} \sqrt{w} 
\; b_{\sigma}\left(\frac{w}{t}\right) e^{i\sigma w } = 
\\
\lim_{\epsilon\to0^+} 
\int_{0}^\infty   \dvol{w} \sqrt{w} 
\; b_{\sigma}(0) e^{i\sigma w -\epsilon w}  +  
\lim_{\epsilon\to0^+}
\int_{0}^\infty   \dvol{w} \sqrt{w} 
\; \left( b_{\sigma}\left(\frac{w}{t}\right)- b_{\sigma}(0) \right) e^{i\sigma w -\epsilon w}. 
\end{gather*}
The limit $\epsilon \to 0$ of the first term gives a finite result. To prove that the second term is bounded we write 
$e^{i\sigma w -\epsilon w} = {(i\sigma-\epsilon)^{-2}}\partial_w^2 (e^{i\sigma w -\epsilon w}-1)$
and we integrate by parts two times ending up with
\begin{gather*}
\int_{0}^\infty   \dvol{w} \sqrt{w} 
\; \left( b_{\sigma}\left(\frac{w}{t}\right)- b_{\sigma}(0) \right) e^{i\sigma w -\epsilon w}
=
\int_{0}^\infty   \dvol{w}\; w^{-3/2}
(e^{i\sigma w -\epsilon w}-1) 
\;  c_{\sigma}\left(\frac{w}{t}\right) 
\end{gather*}
where $c_\sigma$ is another suitable bounded function. Hence, also these second term is bounded by a constant in time.
We can conclude that 
\[
\left|(\omega_2^\beta-\omega_2^\infty)(x; t_y+t, \vettore{y})\right| \leq \frac{C}{t^{3/2}}
\]
combining both estimates we obtain the result in the case $\alpha=0$.

The proof for the case of a generic $\alpha$, can be obtained in a similar way. To estimate $D^{(\alpha)}\omega_2^\beta$, we observe that when the derivatives are applied to the factor in front of the Bessel function, the decaying for large $t$ is improved. Furthermore, the recursive relations of Bessel functions and their asymptotic properties imply  
\[
\frac{d}{dx}K_n(x) = \frac{n}{x}K_n(x)-K_{n+1}(x),\qquad |K_n(y)|\leq  \frac{c_n}{\sqrt{|y|}}, \qquad y \gg n.
\]
Hence, the decaying rate of $D^{(\alpha)}\omega_2^\beta$ for large $t$ is not worse then that of the case $\alpha=0$. 
To estimate the contribution $D^{(\alpha)}(\omega_2^\beta-\omega_2^\infty)$, we apply the derivatives to 
$\frac{\sin(\sqrt{E^2-m^2} |\vettore{x}|)}{\sqrt{E^2-m^2}|\vettore{x}|}e^{i{\sigma} E t}$, and afterwards we proceed in the same way as for the case $\alpha=0$. 
Again, the decay in $t$ cannot be worse then that of the case $\alpha=0$. 
\end{proof}

\begin{proposition}\label{pr-decay-functional-derivatives}
Let $A,B\in\mathscr{F}_{\mu c}$. Let $\omega_2^\beta$ the two-point function \eqref{eq-KMS-free}. It holds that
\[
\left|\left\langle \omega_2^\beta,\frac{\delta^2}{\delta\phi_1\delta\phi_2} \alpha_{t_1}(A)\otimes \alpha_{t_2}(B)\right\rangle\right| \leq \frac{C}{(|t_1-t_2|+r)^{3/2}}
\]
for every $t_1,t_2$ and for some constants $C,r$ which may depend on $A$ and $B$.
\end{proposition}
\begin{proof}
The elements $A$ and $B$ in $\mathscr{F}_{\mu c}$ have compact support by definition. It exits a compact set $\mathcal{C}\subset M$ which contains the supports of both $A$ and $B$. Hence, the action of time translation is such that, the support of $\alpha_t(A)\subset \mathcal{C}_t:=\{(\tau,\vettore{x})\in M | (\tau-t,\vettore{x}) \in \mathcal{C} \}$ which is equal to $\mathcal{C}$ translated of time $t$. 
Since $\mathcal{C}$ is of compact support, if $|t_1-t_2|>d$ with $d$ sufficiently large, there are no lightlike geodesics intersecting both $\mathcal{C}_{t_1}$ and $\mathcal{C}_{t_2}$. $\omega_2^\beta$ is of Hadamard type, hence, when it is restricted to $\mathcal{C}_{t_1}\times\mathcal{C}_{t_2}$
it is described by a smooth function. Since $\frac{\delta^2}{\delta\phi_1\delta\phi_2} \alpha_{t_1}(A)\otimes \alpha_{t_2}(B)$ is a distribution, by continuity we have that, for every $f\in \mathcal{E}(\mathcal{C}_{t_1}\times \mathcal{C}_{t_2})$
\[
\left|\left\langle f,\frac{\delta^2}{\delta\phi_1\delta\phi_2} \alpha_{t_1}(A)\otimes \alpha_{t_2}(B)\right\rangle\right| \leq C_1 \sum_{|\alpha|\leq N}\|D^{\alpha}f\|_\infty 
\]  
where $\alpha$ is a multiindex while $N$ and $C_1$ are two fixed constants. Hence, Proposition \ref{pr-decay-omega-beta} implies that 
\begin{equation}\label{eq:stima-1}
\left|\left\langle \omega_2^\beta,\frac{\delta^2}{\delta\phi_1\delta\phi_2} \alpha_{t_1}(A)\otimes \alpha_{t_2}(B)\right\rangle\right| \leq \frac{C_2}{(|t_1-t_2|)^{3/2}}
\end{equation}
for every $|t_1-t_2| > d$ and for some constant $C_2$.  For every $|t_1-t_2|\leq d$, the product of the distributions $\omega_2^\beta$ and $\frac{\delta^2}{\delta\phi_1\delta\phi_2}\alpha_{t_1}(A)\otimes \alpha_{t_2}(B)$ is well defined because, $\omega_2^\beta$ is an Hadamard two-point function and $A$ and $B$ are in $\mathscr{F}_{\mu c}$, thus the H\"ormander criterion for multiplication of distribution is satisfied. Furthermore, from the support properties of $A$ and $B$ we have that $\omega_2^\beta \cdot \frac{\delta^2}{\delta\phi_1\delta\phi_2}\alpha_{t_1}(A)\otimes \alpha_{t_2}(B)$ is in $\mathcal{E}'(M^2)$. By continuity we have that
\begin{equation}\label{eq:stima-2}
\left|\left\langle \omega_2^\beta,\frac{\delta^2}{\delta\phi_1\delta\phi_2} \alpha_{t_1}(A)\otimes \alpha_{t_2}(B)\right\rangle\right| \leq C_3, 
\end{equation}
for every $|t_1-t_2| \leq d$.
Combining \eqref{eq:stima-1} and \eqref{eq:stima-2} we have the result.
\end{proof}


\begin{thebibliography}{999}

\bibitem[Al90]{Altherr}
T.~Altherr,
\emph{``Infrared problem in $g\varphi^4$ theory at finite temperature,''} 
Phys.\ Lett.\ {\bf B 238} (24), (1990) 360-366. 

\bibitem[Ar73]{Araki}
H.~Araki,
\emph{``Relative Hamiltonian for faithful normal states of a von Neumann algebra,''} 
Publ.\ Res.\ Inst.\ Math.\ Sci.\ {\bf 9}(1), (1973) 165-209.

\bibitem[BB02]{BB02}
J.~Bros and D.~Buchholz,
\emph{``Asymptotic dynamics of thermal quantum fields,''}
Nucl.\ Phys.\ B {\bf 627} (2002) 289.

\bibitem[BKR78]{BrKiRo}
O.~Bratteli, A.~Kishimoto,  D.W.~Robinson, 
\emph{``Stability properties and the KMS condition,''} 
Commun.\ Math.\ Phys.\ {\bf 61}, (1978) 209-238. 

\bibitem[BR97]{BR}
O.~Bratteli,  D.W.~Robinson, 
\emph{``Operator algebras and quantum statistical mechanics 2,''} 
Springer, Berlin (1997).

\bibitem[BDF09]{BDF}
  R.~Brunetti, M.~Duetsch and K.~Fredenhagen,
  \emph{``Perturbative Algebraic Quantum Field Theory and the Renormalization
  Groups,''}
Adv.\ Theor.\  Math.\  Phys.\ {\bf 13}, (2009) 1541-1599.

\bibitem[BF00]{BF00}
  R.~Brunetti and K.~Fredenhagen,
\emph{``Microlocal analysis and interacting quantum field theories:
  Renormalization on physical backgrounds,''}
  Commun.\ Math.\ Phys.\  {\bf 208} (2000) 623.

\bibitem[BF09]{BF09}
  R.~Brunetti and K.~Fredenhagen,
  {\it ``Quantum Field Theory on Curved Backgrounds,''}
  in Lecture Notes in Physics 786, ed. Springer (2009), pp. 129-155 Chapter 5.
  
\bibitem[BFK95]{BFK}
  R.~Brunetti, K.~Fredenhagen and M.~K\"ohler,
  \emph{``The microlocal spectrum condition and Wick polynomials of free fields on
  curved spacetimes,''}
  Commun.\ Math.\ Phys.\  {\bf 180}, (1996) 633.
 
\bibitem[BFV03]{BFV}
  R.~Brunetti, K.~Fredenhagen and R.~Verch,
\emph{``The generally covariant locality principle: A new paradigm for local quantum physics,''}
  Commun.\ Math.\ Phys.\  {\bf 237}, (2003) 31.

\bibitem[CF09]{CF}
B.~Chilian, K.~Fredenhagen,
\emph{``The time-slice axiom in perturbative quantum field theory on globally hyperbolic spacetimes,''}
Commun.\ Math.\ Phys.\ {\bf 287}, (2009) 513-522. 

\bibitem[DHP16]{DHP}
N.~Drago, T.-P.~Hack, N.~Pinamonti, 
\emph{``The generalised principle of perturbative agreement and the thermal mass,''}
To appear on Ann.\ Henri\ Poincar\'e (2016) [arXiv:1502.02705]

\bibitem[DF04]{DF}
M.~D\"utsch and K.~Fredenhagen,
\emph{``Causal perturbation theory in terms of retarded products, and a proof of the action ward identity,''}
Rev.\ Math.\ Phys.\ {\bf 16}, (2004) 1291.

\bibitem[EG73]{EG}
H.~Epstein and V.~Glaser, 
\emph{``The role of locality in perturbation theory, ''}
Ann.\ Inst.\ Henri\ Poincar\'e Section
A, vol. XIX, n.3, {\bf 211} (1973)

\bibitem[FL14]{FredenhagenLindner}
K.~Fredenhagen, F.~Lindner,
\emph{``Construction of KMS States in Perturbative QFT and Renormalized Hamiltonian Dynamics,''}
Commun.\ Math.\ Phys.\  {\bf 332},  (2014) 895.

\bibitem[FR12]{FredenhagenRejzner}
K.~Fredenhagen, K.~Rejzner,
\emph{``Perturbative algebraic quantum field theory,''}
arXiv:1208.1428 [math-ph] (2012).

\bibitem[FR14]{FredenhagenRejzner2}
K.~Fredenhagen, K.~Rejzner,
\emph{``QFT on curved spacetimes: axiomatic framework and examples,''}
arXiv:1412.5125 [math-ph] (2014).

\bibitem[GR07]{Grad}
I. S. Gradshteyn, I. M. Ryzhik, \emph{Table of Integrals, Series, and Products}, Seventh Edition,
Academic Press, (2007).

\bibitem[Ha92]{Haag}
R.~Haag,
\emph{``Local Quantum Physics. Fields Particles Algebras.''}
Text and Monographs in Physics. Springer-Verlag, Berlin (1992).

\bibitem[HKT74]{HKT}
R.~Haag, D.~Kastler, E.~B.~Trych-Pohlmeyer, 
\emph{``Stability and Equilibrium States,''}
Commun.\ Math.\ Phys.\ {\bf 38}, (1974) 173-193.

\bibitem[HHW67]{HHW}
R.~Haag, N.M.~Hugenholtz, M.~Winnink, 
\emph{``On the equilibrium states in quantum statistical mechanics,''}
Commun.\ Math.\ Phys.\ {\bf 5}, (1967) 215-236.

\bibitem[HW01]{HW01}
  S.~Hollands and R.~M.~Wald,
  \emph{``Local Wick polynomials and time ordered products of quantum fields in curved spacetime,''}
  Commun.\ Math.\ Phys.\  {\bf 223},  (2001) 289.

\bibitem[HW02]{HW02}
  S.~Hollands and R.~M.~Wald,
\emph{``Existence of local covariant time ordered products of quantum fields in curved spacetime,''}
  Commun.\ Math.\ Phys.\  {\bf 231}, (2002) 309 .

\bibitem[HW03]{HW03}
  S.~Hollands and R.~M.~Wald,
\emph{``On the Renormalization Group in Curved Spacetime,''}
  Commun.\ Math.\ Phys.\  {\bf 237}, (2003) 123-160.

\bibitem[JP01]{JP01}
V.~Jak{\v s}i\'c, C.-A.~Pillet,
\emph{``On entropy production in quantum statistical mechanic,''}
Commun.\ Math.\ Phys.\ {\bf 217}, (2001) 285.

\bibitem[JP02]{JP02}
V.~Jak{\v s}i\'c, C.-A.~Pillet,
\emph{``Non-Equilibrium Steady States of Finite Quantum Systems Coupled to Thermal Reservoirs,''}
Commun.\ Math.\ Phys.\ {\bf 226}, (2002) 131-162.

\bibitem[LW97]{Landsman}
N.~P.~Landsman, C.~G.~van Weert,
\emph{``Real and imaginary time field theory at finite temperature and density''} 
Phys.\ Rep.\  {\bf 145}  (1987) 141.

\bibitem[Le00]{LeBellac} 
M.~Le Bellac, 
\emph{``Thermal Field Theory,''}
Cambridge University Press, Cambridge (2000).

\bibitem[Li13]{Lindner}
F.~Lindner,
\emph{``Perturbative Algebraic Quantum Field Theory at Finite Temperature.''}
PhD thesis, University of Hamburg (2013). 

\bibitem[OHI88] {Oj0}
I.~Ojima, 
H.~Hasegawa, 
M.~Ichiyanagi,
\emph{``Entropy production and its positivity in nonlinear response theory of quantum dynamical systems,''} 
J. Stat. Phys. {\bf 50}, (1988) 633.


\bibitem[Oj89]{Oj1}
I.~Ojima, 
\emph{``Entropy production and non-equilibrium stationarity in quantum dynamical systems: physical meaning of van Hove limit,''}
J.\ Stat.\ Phys. {\bf 56},  (1989) 203.

\bibitem[Oj91]{Oj2}
I.~Ojima, 
\emph{``Entropy production and non-equilibrium stationarity in quantum dynamical systems''.} 
In: Proceedings of international workshop on quantum aspects of optical communications. Lecture
Notes in Physics 378, 164. Berlin: Springer-Verlag, (1991).

\bibitem[Ra96]{Radzikowski}
  M.~J.~Radzikowski,
\emph{``Micro-Local Approach To The Hadamard Condition In Quantum Field Theory On Curved Space-Time,''}
  Commun.\ Math.\ Phys.\  {\bf 179} (1996) 529.

\bibitem[Ro73]{Robinson}
D.W.~Robinson, 
\emph{``Return to equilibrium,''} 
Commun.\ Math.\ Phys.\ {\bf 31}, (1973) 171-189. 

\bibitem[Ru00]{Ru00}
D.Ruelle,
\textit{Natural nonequilibrium states in quantum statistical mechanics},
J. Statistical Phys. 98,57-75(2000).

\bibitem[St71]{Steinmann}
O.~Steinmann,
\emph{``Perturbation Expansions in Axiomatic Field Theory,''}
Lect. Notes in Phys. {\bf 11}. Berlin: Springer-Verlag, 1971

\bibitem[St95]{SteinmannThermal}
O.~Steinmann,
\emph{``Perturbative quantum field theory at positive temperature: an axiomatic approach.,''}
Commun.\ Math.\ Phys.\ {\bf 170}, (1995) 405-416. 

\end{thebibliography}
\end{document}